\pdfoutput=1
\documentclass[11pt,letterpaper]{article}

\usepackage{hyperref}
\hypersetup{bookmarks, breaklinks, urlcolor=blue, colorlinks, citecolor=green!50!black, linkcolor=blue}
\usepackage[letterpaper, left=1in, right=1in, top=0.9in, bottom=0.9in]{geometry}
\usepackage[utf8]{inputenc}
\usepackage[american]{babel}
\usepackage[normalem]{ulem}
\usepackage{amsmath, amssymb, cases, amsthm}
\usepackage{thmtools}
\usepackage[shortlabels]{enumitem}
\usepackage{mdframed}
\usepackage{bbm}
\usepackage{bm}
\usepackage{microtype}
\usepackage{xcolor}
\usepackage{makecell}
\usepackage{mathtools}
\usepackage{float}
\usepackage{modletters}
\usepackage{comment}
\usepackage{multirow}
\usepackage[capitalize,noabbrev]{cleveref}
\usepackage{graphics}

\usepackage{footnotehyper}
\makesavenoteenv{tabular}

\usepackage{pifont}
\newcommand{\cmark}{\ding{51}}%
\newcommand{\xmark}{\ding{55}}%

\declaretheorem[numberwithin=section,refname={Theorem,Theorems},Refname={Theorem,Theorems}]{theorem}

\declaretheorem[numberlike=theorem]{lemma}

\declaretheorem[numberlike=theorem]{corollary}
\declaretheorem[numberlike=theorem,style=definition]{definition}
\declaretheorem[numberlike=theorem]{claim}
\declaretheorem[numberlike=theorem,style=remark]{remark}
\declaretheorem[numberlike=theorem,refname={Fact,Facts},Refname={Fact,Facts},name={Fact}]{fact}

\declaretheorem[numberlike=theorem, refname={Observation,Observations},Refname={Observation,Observations},name={Observation}]{observation}

\def\final{1}  
\def\iflong{\iffalse}
\ifnum\final=0  
\newcommand{\yonggang}[1]{{\color{blue}[{\tiny Yonggang: \bf #1}]\marginpar{*}}}
\newcommand{\danupon}[1]{{\color{red}[{\tiny Danupon: \bf #1}]\marginpar{\color{red}*}}}
\newcommand{\sagnik}[1]{{\color{green!50!black}[{\tiny Sagnik: \bf #1}]\marginpar{\color{green!50!black}*}}}
\newcommand{\todo}[1]{{\color{red}[{\tiny TODO: \bf #1}]\marginpar{\color{red}*}}}
\newcommand{\yuval}[1]{{\bf \color{red!50!black} YUVAL: #1}}
\newcommand{\jan}[1]{{\bf \color{green!50!black} Jan: #1}}
\newcommand{\blikstad}[1]{\textup{\color{magenta} [\textbf{Joakim}: #1]}}
\newcommand{\tawei}[1]{{\color{blue} [{\bf Ta-Wei:} #1]}}
\newcommand{\TODO}[1]{{\color{blue!50!black} [{\bf Todo:} #1]}}
\else 
\newcommand{\yonggang}[1]{}
\newcommand{\danupon}[1]{}
\newcommand{\sagnik}[1]{}
\newcommand{\todo}[1]{}
\newcommand{\yuval}[1]{}
\newcommand{\jan}[1]{}
\newcommand{\blikstad}[1]{}
\newcommand{\tawei}[1]{}
\newcommand{\TODO}[1]{}
\fi  
\usepackage[ruled,vlined,linesnumbered]{algorithm2e}

\newcommand{\polylog}{\mathrm{polylog}}
\newcommand{\poly}{\mathrm{poly}}

\newcommand{\set}[2][ ]{\{#2 \ifthenelse{\equal{#1}{ }}{ }{~|~#1}\}}

\newcommand{\ith}{i^{\scriptsize \mbox{{\rm th}}}}

\newcommand{\T}{{\mathcal{T}}}

\SetKwComment{Comment}{/* }{ */}

\newcommand{\rank}{\mathsf{rank}}
\newcommand{\sspan}{\mathsf{span}}

\Crefname{algocf}{Algorithm}{Algorithms}

\usepackage[most]{tcolorbox}

\usepackage{xparse}
\usepackage{lipsum}

\makeatletter
\newcommand\footnoteref[1]{\protected@xdef\@thefnmark{\ref{#1}}\@footnotemark}
\makeatother

\renewcommand{\paragraph}[1]{\medskip\noindent{\bf #1}\xspace}

\newcommand{\ground}{U}

\title{Fast Algorithms via Dynamic-Oracle Matroids}

\author{Joakim Blikstad\thanks{KTH Royal Institute of Technology, Stockholm, Sweden, \texttt{blikstad@kth.se}. Supported by the Swedish Research Council (Reg. No. 2019-05622). Work done while at the Max Planck Institute for Informatics, Saarbr\"{u}cken, Germany.} \and Sagnik Mukhopadhyay\thanks{University of Sheffield, United Kingdom, \texttt{s.mukhopadhyay@sheffield.ac.uk}. The second-named author maintains that his role has been limited to the lower bound part.} \and Danupon Nanongkai\thanks{Max Planck Institute for Informatics, Saarbr\"{u}cken, Germany \& KTH, \texttt{danupon@gmail.com}. Partially supported by the Swedish Research Council (Reg. No. 2019-05622). } \and Ta-Wei Tu\thanks{Max Planck Institute for Informatics, Saarbr\"{u}cken, Germany, \texttt{tu.da.wei@gmail.com}.} }
\date{}

\begin{document}
	
	\begin{titlepage}
		\maketitle \pagenumbering{roman}
		
		\begin{abstract}

We initiate the study of matroid problems in a new oracle model called {\em dynamic oracle}. Our algorithms in this model lead to new bounds for some classic problems, and a ``unified'' algorithm whose performance matches previous results developed in various papers for various problems. We also show a lower bound that answers some open problems from a few decades ago. Concretely, our results are as follows.

\paragraph{Improved algorithms for matroid union and disjoint spanning trees.} We show an algorithm with $\tO_k(n+r\sqrt{r})$ dynamic-rank-query and time complexities for the matroid union problem over $k$ matroids, where $n$ is the input size,  $r$ is the output size, and
$\tO_k$ hides $\poly(k, \log(n))$. This implies the following consequences.
{\bf (i)} An improvement over the $\tO_k(n\sqrt{r})$ bound implied by [Chakrabarty-Lee-Sidford-Singla-Wong FOCS'19] for matroid union in the traditional rank-query model.
{\bf (ii)} An $\tO_k\left(|E|+|V|\sqrt{|V|}\right)$-time algorithm for the $k$-disjoint spanning tree problem. This is nearly linear for moderately dense input graphs and improves the 
$\tO_k\left(|V|\sqrt{|E|}\right)$ bounds of Gabow-Westermann [STOC'88] and Gabow [STOC'91]. 
Consequently, this gives improved bounds for, e.g., Shannon Switching Game and Graph Irreducibility.

\paragraph{Matroid intersection.} We show a matroid intersection algorithm with $\tO(n\sqrt{r})$ dynamic-rank-query and time complexities. This implies new bounds for some problems (e.g. maximum forest with deadlines) and bounds that match the classic ones obtained in various papers for various problems, e.g.  
colorful spanning tree [Gabow-Stallmann ICALP'85], 
graphic matroid intersection [Gabow-Xu FOCS'89], 
simple job scheduling matroid intersection [Xu-Gabow ISAAC'94],
and Hopcroft-Karp combinatorial bipartite matching.
More importantly, this is done via a ``unified'' algorithm in the sense that an improvement over our dynamic-rank-query algorithm would 
imply improved bounds for {\em all} the above problems simultaneously.

\paragraph{Lower bounds.} We show simple super-linear ($\Omega(n\log n)$) query lower bounds for matroid intersection and union problems in our dynamic-rank-oracle and the traditional independence-query models; the latter improves the previous $\log_2(3)n - o(n)$ bound by Harvey [SODA'08] and answers an open problem raised by, e.g., Welsh [1976] and CLSSW [FOCS'19].
\end{abstract}

		\setcounter{tocdepth}{3}
		\newpage
		\tableofcontents
		\newpage
	\end{titlepage}
	
	\newpage
	\pagenumbering{arabic}

\section{Introduction} \label{sec:intro}

Via reductions to the max-flow and min-cost flow problems, exciting progress has been recently made for many graph problems such as maximum matching, vertex connectivity, directed cut, and Gomory-Hu trees \cite{Madry13,LeeS14,Madry16,
BrandLNPSSSW20,KathuriaLS20,LiP20,
AbboudKT21stoc,BrandLLSS0W21-maxflow,LiNPSY21,
AbboudKT21focs,AxiotisMV21,Cen0NPSQ21,0006PS21,GaoLP21,
AbboudKT22,CenLP22,
BrandGJLLPS22,
AbboudK0PST22,ChenKLPGS22,
CenHLP23}. 
However, many basic problems still witness no progress since a few decades ago. These problems include $k$-disjoint spanning trees \cite{gabow1988forests,Gabow91}, colorful spanning tree \cite{GabowS85}, arboricity \cite{Gabow95}, spanning tree packing \cite{gabow1988forests}, graphic matroid intersection \cite{GabowS85,GabowX89}, and simple job scheduling matroid intersection \cite{XuG94}. 
For example, in the {\em $k$-disjoint spanning trees problem} \cite[Chapter 51]{schrijver2003}, 
we want to find $k$ edge-disjoint spanning trees in a given input graph $G=(V,E)$. When $k=1$, this is the spanning tree problem and can be solved in linear time. 
For higher values of $k$, the best runtime remains 
$\tO(k^{3/2}|V|\sqrt{|E|})$-time algorithm from around 1990 \cite{gabow1988forests}\footnote{The stated bound was due to Gabow and Westermann \cite{gabow1988forests}, which was usually referred to as the state of the art (e.g.\ in \cite{schrijver2003,BlumenstockF20,Quanrud23,HarrisSV21}). Note that Gabow \cite{Gabow91} announced an improved bound of $O(k|V|\sqrt{|E|+k|V|\log(|V|)})$ but this bound was later removed from the journal version of the paper. After our paper was accepted to STOC'23, we are aware of the paper by Karger \cite{Karger98} that also studies this problem. The paper claims the runtime of $\tO(|E|+k^{5/2}|V|^3)$ (via matroid union), but it seems that the technique in the paper may imply $\tO(|E|+\poly(k)|V|^{3/2})$ runtime when combined with \cite{gabow1988forests}. We discuss a relevant concurrent result \cite{Quanrud23} later in this section.}, which is also the best runtime for its applications such as Shannon Switching Game \cite{Gardner61} and Graph $k$-Irreducibility \cite{Whiteley88rigidity,graver1993combinatorial}. 
No better runtime was known even for the special case of $k=2$.

{\em Can we improve the 
bounds of $k$-disjoint spanning trees and other problems?}
More importantly, since it is very unclear if these problems can be reduced to max-flow or min-cost flow\footnote{For example, the best-known number of max-flow calls to decide whether there are $k$ disjoint spanning trees and to find the $k$ spanning trees are $O(n)$ and $O(n^2)$ respectively.}, 
{\em is there an alternative approach to designing fast algorithms for many problems simultaneously?}
Fortunately, many of the above problems can be modeled as {\em matroid problems}, giving hope that solving matroid problems would solve many of these problems in one shot. Unfortunately, this is not true in the traditional model for matroid problems---even the most efficient algorithm possible for a matroid problem does not 
necessarily give a faster algorithm for any of its special cases. 
We discuss this more below.

\paragraph{Matroid Problems.} A matroid $\cM$ is a pair $(U, \cI)$ where $U$ is a finite set (called the ground set) and $\cI$ is a family of subsets of $U$ (called the independent sets) satisfying some constraints (see \cref{def:matroid}; these constraints are not important in the following discussion).
Since $\cI$ can be very large, problems on matroid $\cM$ are usually modeled with {\em oracles} that answer {\em queries}. Given a set $S\subseteq U$, {\em independence queries} ask if $S\in \cI$ and {\em rank queries} ask for the value of $\max_{I\in \cI, I\subseteq S} |I|.$
Two textbook examples of matroid problems are {\em matroid intersection} and {\em union}\footnote{\emph{Matroid union} is also sometimes called \emph{matroid sum}.} (e.g., \cite[Chapters~41-42]{schrijver2003}).  
We will also consider the special case of matroid union called {\em $k$-fold matroid union}.

\begin{definition}[Matroid intersection and ($k$-fold) matroid union]\label{def:intro:matroid intersection union}
(I) In matroid intersection, we are given two matroids $(U, \cI_1)$ and $ (U, \cI_2)$ and want to find a set of maximum size in $\cI_1\cap \cI_2.$ 
(II) In matroid union, we are given $k$ matroids $(U_1, \cI_1), (U_2, \cI_2), \ldots, (U_k, \cI_k)$, and want to find the set $S_1\cup S_2 \cup \cdots \cup S_k$, where $S_i \in \cI_i$ for every $i$, of maximum size. 
(III) Matroid union in the special case where $U_1=U_2=\cdots = U_k$ and $\cI_1=\cI_2=\cdots = \cI_k$ is called  $k$-fold matroid union.  

{\em Notations:} Throughout, for problems over matroids  $(U_1, \cI_1), (U_2, \cI_2), \ldots, (U_k, \cI_k)$, we define  $n:=\max_i{|U_i|}$ and $r:=\max_{i}\max_{S \in \cI_i}|S|$. \qed
\end{definition}

Matroid problems are powerful abstractions that can model many fundamental problems. 
For example, the $2$-disjoint spanning tree problem can be modeled as a $2$-fold matroid union problem: 
\begin{tcolorbox}[breakable,boxrule=0pt,frame hidden,sharp corners,enhanced,borderline west={.5 pt}{0pt}{red},colback=white]
\vspace{-.2cm}
Given a graph $G=(V, E)$, let $\cM=(U, \cI)$ be the corresponding {\em graphic matroids}, i.e. $U=E$ and $S\subseteq E$ is in $\cI$ if it is a forest in $G$. (It is a standard fact that such an $\cM$ is a matroid.) 
The $2$-fold matroid union problem with input $\cM$ is a problem of finding two forests $F_1\subseteq E$ and $F_2\subseteq E$ in $G$ that maximizes $|F_1\cup F_2|$. This is known as the {\em $2$-forest} problem which clearly generalizes $2$-disjoint spanning tree  (a $2$-forest algorithm will return two disjoint spanning trees in $G$ if they exist). 
\end{tcolorbox}
Observe that this argument can be generalized to modeling the $k$-disjoint spanning trees problem by $k$-fold matroid union.  
Other problems that can be modeled as matroid union (respectively, matroid intersection) include arboricity,  spanning tree packing, $k$-pseudoforest, and mixed $k$-forest-pseudoforest (respectively, bipartite matching and colorful spanning tree).

The above fact makes matroid problems a {\em unified} approach for showing that many problems, including those mentioned above, can be solved in polynomial time. This is because (i) the matroid union, intersection, and other problems can be solved in polynomial time and rank/independence queries, and (ii) for most problems  queries can be answered in polynomial time. For example, when we model $k$-disjoint spanning trees as $k$-fold graphic matroid union like above, the corresponding rank query is: given a set $S$ of edges, find the size of a spanning forest of $S$. This can be solved in $O(|S|)$ time. 

When it comes to more fine-grained time complexities, such as nearly linear and sub-quadratic time, matroid algorithms in the above model are not very helpful. This is because simulating a matroid algorithm in this model causes too much runtime blow-up. 
For example, even if we can solve $k$-fold matroid union over $(U, \cI)$ in {\em linear} ($O(|U|)$) rank query complexity, it does not necessarily imply that we can solve its special case of $2$-disjoint spanning tree any faster.
This is because each query about a set $S$ of edges needs at least $O(|S|)$ time even to specify $S$, which can be as large as the number of edges in the input graph.
In other words, even a matroid union algorithm with linear complexities may only imply $O(|E|^2)$ time for solving 2-disjoint spanning trees on graphs $G=(V,E)$. 
This is also the case for other problems that can be modeled as matroid union and intersection. 
Because of this, previous works obtained improved bounds by simulating an algorithm for matroid problems and coming up with clever ideas to speed up the simulation for each of these problems one by one (e.g., \cite{GabowT79,RoskindT85,GabowS85,gabow1988forests,FredericksonS89,GabowX89,Gabow91,XuG94}).
It cannot be guaranteed that recent and future improved algorithms for matroid problems (e.g., \cite{chakrabarty2019faster,quadratic2021,blikstad21}) would imply improved bounds for any of these problems.

\paragraph{Dynamic Oracle.}
The main conceptual contribution of this paper is an introduction of a new matroid model called {\em dynamic oracle}
and an observation that, using dynamic algorithms, 
solving a matroid problem efficiently in our model immediately implies efficient algorithms for many problems it can model. 
In contrast to traditional matroids where a query can be made with an arbitrary set $S$, our model only allows queries made by slightly modifying the previous queries.\footnote{The ``cost'' of a query in our dynamic model is the distance (size of the symmetric difference) from some (not necessarily the last) previous query.} More precisely, the dynamic-rank-oracle model, which is the focus of this work, is defined as follows.\footnote{One can also define the dynamic-independence-oracle model where $\textsc{Query}(i)$ returns only the independence of $S_i$.}

\begin{definition}[Dynamic-rank-oracle model]
\label{def:dyn-oracle}
  For a matroid $\cM = (\ground, \cI)$, starting from $S_0 = \emptyset$ and $k = 0$, the algorithm can access the oracle via the following three operations.
  \begin{itemize}[noitemsep]
    \item $\textsc{Insert}(v, i)$: Create a new set $S_{k + 1} := S_i \cup \{v\}$ and increment $k$ by one.
    \item $\textsc{Delete}(v, i)$: Create a new set $S_{k + 1} := S_i \setminus \{v\}$ and increment $k$ by one.
    \item $\textsc{Query}(i)$: Return the rank of $S_i$, i.e., the size of the largest independent subset of $S_i$.
  \end{itemize}
 We say that a matroid algorithm takes $t$ time and dynamic-rank-query complexities if its time complexity and required number of operations are both at most $t$. \qed
 \end{definition}

We emphasize that a query can be obtained from \emph{any} previous query, not just the last one.

\begin{table*}[ht]
    \centering

{\small
    \caption{Examples of implications of dynamic-rank-oracle matroid algorithms. The complexities in the first column are in terms of time and dynamic-rank-query complexities. Notations $n$, $r$, and $k$ are as in \Cref{def:intro:matroid intersection union}.
    In the second column, the $\polylog{n}$ factors are hidden in $\tT$ and
    subpolynomial factors are hidden in $\hat{T}.$ Details are in 
    \cref{sec:packing,appendix:applications}.
    }
    \begin{tabular}{c|l}
        {\bf matroid problems} & 
        \begin{tabular}{l}
        {\bf special cases}
        \end{tabular}
        \\\hline
        \begin{tabular}{c}
        $k$-fold matroid union in $T(n, r, k)$ \\
        \end{tabular} & \begin{tabular}{l}
             \\
             $k$-forest in $\tT(|E|, |V|, k))$ \\
             $k$-pseudoforest in $\tT(|E|, |V|, k))$ \\
             $k$-disjoint spanning tree in $\tT(|E|, |V|, k)$ (randomized) \\
             \phantom{$k$-disjoint spanning tree in} $\hat{T}(|E|, |V|, k)$ (deterministic) \\
             arboricity in $\tT(|E|, |V|, \sqrt{|E|}))$ \\
             tree packing in $\tT(|E|, |V|, |E|/|V|))$ \\
             Shannon Switching Game in $\tT(|E|, |V|, 2))$ \\
             graph $k$-irreducibility in $\tT(|E|, |V|, k))$ \\
             \\
        \end{tabular}\\\hline
        \begin{tabular}{c}
          matroid union in $T(n, r, k)$ \\
        \end{tabular}& \begin{tabular}{l}
          \\
          $(f, p)$-mixed forest-pseudoforest in $\tT(|E|, |V|, f + p))$ \\
          \\
        \end{tabular}\\\hline
        \begin{tabular}{c}
          matroid intersection in $T(n, r)$ \\
        \end{tabular} & \begin{tabular}{l}
          \\
          bipartite matching in $\tT(|E|, |V|)$ \\
          colorful spanning tree in $\tT(|E|, |V|)$\\
          graphic matroid intersection in $\tT(|E|, |V|)$ \\
          simple job scheduling matroid intersection in $\tT(n, r)$ \\
          convex transversal matroid intersection in $\tT(|V|, \mu)$ \\
          \\
        \end{tabular}\\\hline
    \end{tabular}
}
    \label{tab:intro:dyn-oracle-implies-fast-algorithms}
\end{table*}

\begin{observation}[Details in \cref{sec:packing,appendix:applications}]
\label{thm:intro:dyn-oracle implies fast alg}
Algorithms for the $k$-fold matroid union, matroid union, and matroid intersection problems imply algorithms for a number of problems with time complexities shown in \Cref{tab:intro:dyn-oracle-implies-fast-algorithms}.
\end{observation}
\begin{proof}[Proof Idea]
As an example, we sketch the proof that if $k$-fold matroid union can be solved in $T(n, r, k)$ then $k$-disjoint spanning trees can be found in  $T(|E|, |V|, k)\cdot\polylog(|V|)$ time. 
Recall that in the traditional rank-oracle model, the algorithm can ask an oracle for the size of a spanning forest in an arbitrary set of edges $S$, causing $O(|S|)$ time to simulate. 
In our dynamic-rank-oracle model, an algorithm needs to modify some set $S_i$ to the desired set $S$ using the {\sc Insert} and {\sc Delete} operations before asking for the size of a spanning forest in $S$. We can use a spanning forest data structure to keep track of the size of the spanning forest under edge insertions and deletions. This takes $\polylog(|V|)$ time per operation \cite{KapronKM13,GibbKKT15}.\footnote{The dynamic spanning forest algorithms of \cite{KapronKM13,GibbKKT15} are randomized and assume the so-called oblivious adversary (as opposed to, e.g.,   \cite{NanongkaiS17,Wulff-Nilsen17} which work against adaptive adversaries). This is not a problem because we only need to report the size of the spanning forest and not an actual forest. We can also use a deterministic algorithm from \cite{ChuzhoyGLNPS20,NanongkaiSW17} which requires $|V|^{o(1)}$ time per operation.}
So, if $k$-fold matroid union can be solved in $T(n, r, k)$ time and dynamic rank queries, then $k$-disjoint spanning trees can be solved in  $\tO(T(n, r, k))=\tO(T(|E|, |V|, k))$ time, where the equality is because the ground set size is the number of edges ($|U|=|E|$) and the rank $r$ is equal to the size of a spanning forest (thus at most $|V|$).\footnote{Note that we also need a fully-persistent data structure \cite{DriscollSST86,Dietz89} to maintain the whole change history in our argument.}
\end{proof}

Observe that designing efficient algorithms in our dynamic-oracle model is {\em not} easier than in the traditional model: a dynamic-oracle matroid algorithm can be simulated in the traditional model within the same time and query complexities. Naturally, the first challenge of the new model is this question: {\em Can we get matroid algorithms in the new model whose performances match the state-of-the-art algorithms in the traditional model?}
Moreover, for the new model to provide a unified approach to solve many problems simultaneously, one can further ask: {\em Would these new matroid algorithms imply state-of-the-art bounds for many problems?}

\paragraph{Algorithms.}  
In this paper, we provide algorithms in the new model whose complexities not only match those in the traditional model but sometimes even improve them.  These lead to new bounds for some problems and, for other problems, a unified algorithm whose performance matches previous results developed in various papers for various problems.

More precisely, the best time and rank-query complexities for matroid intersection on input $(U, \cI_1)$ and $(U, \cI_2)$ were $\tO(n\sqrt{r})$ by Chakrabarty, Lee, Sidford, Singla, and Wong \cite{chakrabarty2019faster} (improving the previous $\tO(nr)$ bound based on Cunningham's classic algorithm \cite{cunningham1986improved,LeeSW15,nguyen2019note}).
Due to a known reduction, this implies $\tO(k^2\sqrt{k}n\sqrt{r})$ bound for $k$-fold matroid union and matroid union. 
In this paper, we present algorithms in the dynamic-oracle model that imply improved bounds in the traditional model for $k$-fold matroid union and matroid union and match the bounds for matroid intersection.

Here, we only state our dynamic-rank-query complexities as they are the main focus of this paper, and for all the applications we have, answering (and maintaining dynamically) independence queries does not seem to be significantly easier.
Note that we also obtain dynamic-independence-query algorithms that match the state-of-the-art traditional ones~\cite{blikstad21} which we defer to
\cref{appendix:dynamic-matroid-intersection-ind}.
\begin{theorem}  \label{thm:main}
(I)  $k$-fold matroid union over input $(\ground, \cI)$ can be solved in $\tO(n + kr\sqrt{\min(n, kr)} + k \min(n, kr))$ time and dynamic rank queries.
(II) Matroid union over input $(\ground_1, \cI_1), (\ground_2, \cI_2), \ldots, (\ground_k, \cI_k)$ can be solved in $\tO\left(\left(n + r\sqrt{r}\right) \cdot \poly(k)\right)$ time and dynamic rank queries.
(III) Matroid intersection over input $(\ground, \cI_1)$ and $(\ground, \cI_2)$ can be solved in $\tO(n\sqrt{r})$ time and dynamic rank queries.
\end{theorem}

\definecolor{applegreen}{rgb}{0.55, 0.71, 0.0}

\begin{savenotes}
\begin{table*}[ht]
    \centering
    \small
    \caption{Implications of our matroid algorithms stated in \cref{thm:main} in comparison with previous results. Results marked with a {\color{green!50!black}\cmark} improve over the previous ones. Results marked with a {\color{red}\xmark} are worse than the best time bounds. Other results match the currently best-known algorithms up to poly-logarithmic factors. Details can be found in 
    \cref{appendix:applications}.
    }
    \begin{tabular}{c|l|l}
      {\bf problems} & {\bf our bounds} & {\bf state-of-the-art results} \\\hline
      {\tiny (Via $k$-fold matroid union)}& & \\
      $k$-forest\footnote{For $k$-forest and its related graph problems in the table, we can assume that $k \leq |V|$, and thus the $k^2r$ (where $r = \Theta(|V|)$) term in \cref{thm:main} is dominated by the $(kr)^{3/2}$ term.} & $\tO(|E| + (k|V|)^{3/2})$ {\color{green!50!black}\cmark} & $\tO(k^{3/2}|V|\sqrt{|E|})$ \cite{gabow1988forests} \\
      $k$-pseudoforest & $\tO(|E| + (k|V|)^{3/2})$ {\color{red}\xmark} & $|E|^{1 + o(1)}$ \cite{ChenKLPGS22} \\
      $k$-disjoint spanning trees & $\tO(|E| + (k|V|)^{3/2})$ {\color{green!50!black}\cmark} & $\tO(k^{3/2}|V|\sqrt{|E|})$ \cite{gabow1988forests} \\
      arboricity\footnote{Here we use the bound that $\alpha \leq \sqrt{E}$~\cite{Gabow95}.} & $\tO(|E||V|)$ {\color{red}\xmark} & $\tO(|E|^{3/2})$ \cite{Gabow95} \\
      tree packing & $\tO(|E|^{3/2})$ & $\tO(|E|^{3/2})$ \cite{gabow1988forests} \\
      Shannon Switching Game & $\tO(|E| + |V|^{3/2})$ {\color{green!50!black}\cmark} & $\tO(|V|\sqrt{|E|})$ \cite{gabow1988forests} \\
      graph $k$-irreducibility & $\tO(|E| + (k|V|)^{3/2} + k^2|V|)$ {\color{green!50!black}\cmark} & $\tO(k^{3/2}|V|\sqrt{|E|})$ \cite{gabow1988forests} \\ & & \\ \hline
      {\tiny (Via matroid union)} & & \\
      $(f, p)$-mixed forest-pseudoforest & $\tO_{f,p}(|E| + |V|\sqrt{|V|})$ {\color{green!50!black}\cmark} & $\tO((f + p)|V|\sqrt{f|E|})$ \cite{gabow1988forests} \\ & & \\\hline
      {\tiny (Via matroid intersection)} & & \\
      bipartite matching (combinatorial\footnoteref{foot:combinatorial}) & $\tO(|E|\sqrt{|V|})$ & $O(|E|\sqrt{|V|})$ \cite{HopcroftK73} \\
      bipartite matching (continuous) & $\tO(|E|\sqrt{|V|})$ {\color{red}\xmark} & $|E|^{1 + o(1)}$ \cite{ChenKLPGS22} \\
      graphic matroid intersection & $\tO(|E|\sqrt{|V|})$ & $\tO(|E|\sqrt{|V|})$ \cite{GabowX89} \\
      simple job scheduling matroid intersection & $\tO(n\sqrt{r})$ & $\tO(n\sqrt{r})$ \cite{XuG94} \\
      convex transversal matroid~\cite{edmonds1965transversals} intersection & $\tO(|V|\sqrt{\mu})$ & $\tO(|V|\sqrt{\mu})$ \cite{XuG94} \\
      linear matroid intersection\footnote{Our bound is with respect to the current value of $\omega < 2.37286$~\cite{AlmanW21}. If $\omega = 2$, then our bound becomes $\tO(n^{2.5}\sqrt{r})$.} & $\tO(n^{2.529}\sqrt{r})$ {\color{red}\xmark} & $\tO(nr^{\omega - 1})$ \cite{Harvey09} \\
      colorful spanning tree & $\tO(|E|\sqrt{|V|})$ & $\tO(|E|\sqrt{|V|})$ \cite{GabowS85} \\
      maximum forest with deadlines & $\tO(|E|\sqrt{|V|})$ {\color{green!50!black}\cmark} & (no prior work) \\
    \end{tabular}
    \label{tab:intro:implications}
\end{table*}
\end{savenotes}

Combined with \Cref{thm:intro:dyn-oracle implies fast alg}, the above theorem immediately implies fast algorithms for many problems. \Cref{tab:intro:implications} shows some of these problems. 
One of our highlights is the improved bounds for $k$-forest and $k$-disjoint spanning trees. Even for $k=2$, there was no runtime better than the decades-old $\tO(k^{3/2}|V|\sqrt{|E|})$ runtime \cite{gabow1988forests,Gabow91}. Our result improves this to  $\tO(|E| + (k|V|)^{3/2})$.
This is nearly linear for dense input graphs and small $k$.
This also implies a faster runtime for, e.g., Shannon Switching Game (see \cite{shannon1955game,gabow1988forests}) which is a special case of $2$-disjoint spanning trees.

Our matroid intersection algorithm gives a unified approach to achieving time complexities that were previously obtained by various techniques in many papers. 
Thus, improving this algorithm would imply breakthrough runtimes for many of these problems simultaneously.
Moreover, in contrast to the previous approach where matroid algorithms have to be considered for each new problem one by one, our approach has the advantage that it can be easier to derive new bounds.
For example, say we are given a graph $G = (V, E)$, where the edge $e$ will stop functioning after day $d(e)$.
Every day we can ``repair'' one functioning edge. 
Our goal is to make the graph connected in the long run  (an edge will work forever once it has been repaired).
This is the \emph{maximum forest with deadlines} problem.
Formally speaking, the goal is to construct a spanning tree or a forest of the maximum size at the end, by selecting an edge $e$ with $d(e) \geq t$ in the $t^{\scriptsize \mbox{{\rm th}}}$ round.\footnote{It is tempting to believe that we can use a greedy algorithm where we always select an edge $e$ with the smallest $d(e)$ to the solution. The following example shows why this does not work: There are three vertices $V=\{a,b,c,d\}$. Edges $e_1$ and $e_2$ between $a$ and $b$ have $d(e_1)=1$ and $d(e_2)=3$. Edges $e_3=(b,c)$ and $e_4=(c,d)$ have $d(e_3)=d(e_4)=2$.
}
Our result implies a runtime of  $\tO(|E|\sqrt{|V|})$ for this problem. The runtime holds even for the harder case where each edge is also associated with an {\em arrival time} (edges cannot be selected before they arrive).

We also list some problems where our bounds cannot match the best bounds in \Cref{tab:intro:implications}. Improving our matroid algorithms to match these bounds is a very interesting open problem. A particularly interesting case is the maximum bipartite matching problem. Our dynamic-oracle matroid intersection algorithm  implies a runtime that matches the runtime from the best combinatorial algorithm of Hopcroft and Karp \cite{HopcroftK73} which has been recently improved via continuous optimization techniques (e.g., \cite{Madry13,Madry16,CohenMSV17,AxiotisMV20,BrandLNPSSSW20,BrandLLSS0W21-maxflow,ChenKLPGS22}).\footnote{\label{foot:combinatorial}The term ``combinatorial'' is vague and varies in different contexts. Here, an algorithm is ``combinatorial'' if it does not use any of the continuous optimization techniques such as interior-point methods (IPMs).}
There are barriers to using continuous optimization even to solve some special cases of matroid intersection (e.g. colorful spanning tree's linear program requires exponentially many constraints). Thus, improving our matroid intersection algorithm requires either a new way to use continuous optimization techniques or a breakthrough idea in designing combinatorial algorithms that would improve the Hopcroft-Karp algorithm.

\paragraph{Lower Bounds.} 
Another advantage of our dynamic-oracle matroid model is that it can be easier to prove lower bounds than the traditional model. As a showcase, we show a simple super-linear rank-query lower bound in our new model. In fact, our argument also implies the first super-linear independence-query lower bound in the traditional model.
The latter result might be of independent interest. 
\begin{theorem}\label{thm:intro:lowerbound}
(I) Any deterministic algorithms require $\Omega(n\log n)$ dynamic rank queries to solve the matroid union and matroid intersection problems.
(II) Any deterministic algorithms require $\Omega(n\log n)$ (traditional) independence queries to solve the matroid union and matroid intersection problems.
\end{theorem}
Our first lower bound suggests that the dynamic-oracle model might at best give nearly linear (and not linear) time algorithms.
Prior to this paper, only a $\log_2(3)n - o(n)$ independence-query lower bound for deterministic algorithms was known for (traditional) independence queries, due to Harvey \cite{harvey2008matroid}.\footnote{To the best of our knowledge, this lower bound does not hold for rank queries.}
Our lower bound in the traditional model improves this decade-old bound. 
Moreover, showing super-linear independence-query lower bounds in the traditional model for matroid intersection is a long-standing open problem considered since 1976 (e.g. \cite{Welsh1976matroid_book,chakrabarty2019faster}).\footnote{As noted by Harvey, Welsh asked about the number of queries
needed to solve the matroid partition problem, which is equivalent to matroid union and intersection.} 
Our lower bound in the traditional model answers this open problem for deterministic algorithms. The case of randomized algorithms would be resolved too if an $\omega(|V|)$ lower bound was proved for the communication complexity for computing connectivity of an input graph $G=(V,E)$. (It was conjectured to be $\Omega(n\log n)$ in \cite{ApersEGLMN22}.)

\paragraph{Independent Work.}
Concurrently and independently to our work, Quanrud~\cite{Quanrud23} also studied the $k$-fold matroid union and related problems and obtained a similar running time of $\tO(\min(n, kr)^{3/2})$ to ours in the \emph{traditional independence-oracle model}.
By specializing the algorithm to graphic matroids, Quanrud also obtained an $\tO(\min(|E|, k|V|)^{3/2})$ algorithm for $k$-disjoint spanning tree. The techniques used in these two works are different, however, and our main contribution remains the introduction of \emph{dynamic} oracles and efficient matroid algorithms in this model.

\subsection{Techniques}
In this section we briefly discuss our technical contributions. For a more in-depth overview of our algorithms, see the technical overview (\cref{sec:overview}).

\paragraph{Exchange Graph \& Blocking Flow.} Our algorithms and lower bounds are based on the notion of finding \emph{augmenting paths} in the \emph{exchange graph}, due to \cite{edmonds1970submodular,Lawler75,aignerD}.
Given a common independent set $S\in \cI_1\cap \cI_2$, the exchange graph $G(S)$ is a directed graph where finding an $(s,t)$-path corresponds to increasing the size of $S$ by one. Starting with the work of Cunningham \cite{cunningham1986improved}, modern matroid intersection algorithms (including the state-of-the-art \cite{chakrabarty2019faster,blikstad21}) are based on a ``Blocking Flow'' idea inspired by the Hopcroft-Karp's \cite{HopcroftK73} bipartite matching and Dinic's \cite{dinic1970algorithm}  max-flow algorithms.

\paragraph{Matroid Intersection with Dynamic Oracle.} 
Our matroid intersection algorithms are implementations of the state-of-the-art $\tO(n\sqrt{r})$ rank-query algorithm of \cite{chakrabarty2019faster} and
the $\tO(n r^{3/4})$ independence-query algorithm of \cite{blikstad21}. Our contribution here is to show that versions of them can be implemented also in the \emph{dynamic}-oracle model. 

These algorithms explore the exchange graph efficiently in the classic non-dynamic models by performing binary searches with the oracle queries to find useful edges. However, such a binary search is very expensive in the dynamic-oracle model (as the queries differ by a lot): a single such binary search might cost up to $O(n)$ in the dynamic-oracle model instead of just $O(\log n)$.

Our contribution is to design a binary-tree data structure that supports finding these useful edges efficiently also in the dynamic-oracle model.
Note that after each augmentation the underlying exchange graph changes, so the data structure must also support these dynamic updates efficiently. Some updates can just be propagated up the tree, while others we handle by batching them and rebuilding the tree periodically.
We also rely on a structural result ``Augmenting Sets'' by \cite{chakrabarty2019faster} which states that the updates to the exchange graph are local, which helps us reduce the number of updates we need to make to our data structure, and achieve the final time bound.

\paragraph{Matroid Union with Dynamic Oracle.} 
Our $\tO(n + r\sqrt{r})$ matroid union algorithm with dynamic rank oracle is based on our $\tO(n \sqrt{r})$ matroid intersection algorithm (indeed, matroid union is a special case of matroid intersection).
We are able to obtain a more efficient algorithm by 
taking advantage of the additional structure of the exchange graph in the case of matroid union.
The main idea is to run the blocking flow algorithm only on a dynamically-changing subgraph of size $\Theta(r)$, instead of on the full exchange graph of size $\Theta(n)$.

A crucial observation is that all but $O(r)$ elements will be directly connected to the source vertex~$s$. To ``sparsify'' this first layer in the breadth-first-search tree, we argue that one only needs to consider a basis of it (this basis will have size at most $r$ as opposed to $n$). After an augmentation, this first layer changes, so we design a \emph{dynamic algorithm to maintain a basis of a matroid}\footnote{For example, maintaining a spanning forest in a dynamically changing graph.}, with $\tO(\sqrt{r})$ update time and $O(n)$ pre-computation.
Our algorithm to maintain this basis dynamically is inspired by the
 dynamic minimum spanning tree algorithm of \cite{Frederickson85} ($O(\sqrt{|E|})$ update time), in combination with the sparsification trick of \cite{EppsteinGIN97} ($\tO(\sqrt{|V|})$ update time).
 We believe that our dynamic algorithm to maintain a (min-weight) basis of a matroid might also be of independent interest.

\paragraph{Lower Bounds.} Our super-linear $\Omega(n \log n)$ query lower bound comes from studying the communication complexity of matroid intersection. The matroids $\cM_1$ and $\cM_2$ are given to two parties Alice and Bob respectively and they are asked to solve the matroid intersection problem using as few bits of communication between them. We show that even if Alice and Bob know some common independent set $S\in \cI_1\cap \cI_2$, they need to communicate $\Omega(n\log n)$ bits to see if $S$ is optimal. Essentially, they need to determine if there is an augmenting path in the exchange graph. Using a class of matroids called \emph{gammoids} (see e.g.~\cite{perfect1968applications,mason1972class}), we show a reduction from the $(s,t)$-connectivity problem which has a deterministic $\Omega(n\log n)$ communication lower bound \cite{HajnalMT88}.

\subsection{Organization}

The rest of the paper is organized as follows.
We first give a high-level overview of how we obtain our algorithms in \cref{sec:overview}.
In \cref{sec:prelim}, we provide the necessary preliminaries.
We then construct the binary search tree data structure in \cref{sec:bst}, followed in \cref{sec:matroid-intersection} by how to use it to implement our $\tO(n\sqrt{r})$ matroid intersection algorithm in the new dynamic-rank-oracle model (the dynamic-\emph{independence}-oracle algorithm is in \cref{appendix:dynamic-matroid-intersection-ind}).
In \cref{sec:decremental-basis} we describe our data structure to maintain a basis of a matroid dynamically, and then we use this in our $\tO_k(n+r\sqrt{r})$ matroid union algorithm in \cref{sec:matroid-union} (the special case of $k$-fold matroid union is in \cref{appendix:matroid-union-fold}).
We show our super-linear lower bound in \cref{sec:lowerbound}.
We end our paper with a discussion of open problems in \cref{sec:openproblems}.
In \cref{appendix:applications} we mention how to implement different matroids oracles in the dynamic-oracle model, and discuss some problems we can solve with our algorithms.

\section{Technical Overview of Algorithms}
\label{sec:overview}

\subsection{The Blocking-Flow Framework}
In this section, we give a high-level overview of our algorithms.
We will focus on the dynamic-rank-oracle model (\cref{def:dyn-oracle}), and sketch how to efficiently implement the ``blocking flow''\footnote{Similar to the Hopcroft-Karp's \cite{HopcroftK73} bipartite matching and Dinic's \cite{dinic1970algorithm} maximum flow algorithms.} matroid intersection algorithms of \cite{GabowS85,cunningham1986improved,chakrabarty2019faster,nguyen2019note} in this model.
As such, we briefly recap how the $\tO(n\sqrt{r})$ rank-query (in the traditional oracle model) algorithm of \cite{chakrabarty2019faster} works first, and then explain how to implement their framework in the new dynamic oracle model with the same cost.

Their algorithm, like most of the matroid intersection algorithms, is based on repeatedly finding \emph{augmenting paths} in \emph{exchange graphs} (see \cref{sec:prelim} for a definition). Say we have already found some common independent set $S\in \cI_1\cap \cI_2$ (we start with $S = \emptyset$). Then the \emph{exchange graph} $G(S)$ is a directed bipartite graph in which finding an $(s,t)$-path exactly corresponds to increasing the size of $S$ by one.
According to Cunningham's blocking-flow argument \cite{cunningham1986improved}, if we always augment along the \emph{shortest} augmenting path, the lengths of such augmenting paths are non-decreasing.
Moreover, if the length of the shortest augmenting path in $G(S)$ is at least $1/\epsilon$, then the size of the current common independent set $S$ must be at least $\left(1 - O(\epsilon)\right) r$ (i.e.\ it is only $O(\epsilon r)$ away from optimal).
Thus, the ``blocking flow''-style algorithms consists of two stages:
\begin{enumerate}
\item
In the first stage, they obtain a $(1 - \epsilon)$-approximate solution by finding augmenting paths until their lengths become more than $O(1/\epsilon)$. This is done by running in phases, where in phase~$i$ they eliminate all augmenting paths of length $2i$ by finding a so-called ``blocking flow''---a maximal (not necessarily maximum) collection of compatible augmenting paths. Each such phase can be implemented using only $\tO(n)$ rank queries, as shown in \cite{chakrabarty2019faster}. This means that the first stage needs a total of $\tO(n/\epsilon)$ rank queries (in the classic non-dynamic model).

\item
In the second stage, they find the remaining $O(\epsilon r)$ augmenting paths one at a time. Each such augmentation can be found in $\tO(n)$ rank queries, for a total of $\tO(\epsilon nr)$ queries for this stage.
\end{enumerate}
Using $\epsilon = 1/\sqrt{r}$, \cite{chakrabarty2019faster} obtains their $\tO(n \sqrt{r})$ rank-query exact matroid intersection algorithm.
The crux of how to implement the stages efficiently is a binary search trick to explore useful edges of the exchange graph quickly (for e.g.\ to implement a breadth-first-search on the graph). The exchange graph can have up to $\Theta(nr)$ edges in total, but it is not necessary to find all of them. We will argue that this binary search trick (which issues queries far away from each other) can still be implemented in the \emph{dynamic-oracle model}, with the use of some data structures.

\subsection{Matroid Intersection}

\paragraph{Binary Search Tree.}
The crux of why a breadth-first-search (BFS) and augmenting path searching can be implemented efficiently (in terms of the number of traditional queries) in \cite{chakrabarty2019faster} is that they show how to, for $S \in \cI$, $u \in S$, and $X \subseteq \bar{S}$, discover an element $x \in X$ with $(S \setminus \{u\}) \cup \{x\} \in \cI$ in $O(\log{n})$ rank queries using binary search (such a pair $(u,x)$ is called an \emph{exchange pair}, and corresponds to an edge in the exchange graph).
The idea is that such an $x$ exists in $X$ if and only if $\rank((S \setminus \{u\}) \cup X) \geq |S|$.
Thus, we can do a binary search over $X$: we split $X$ into two equally-sized subsets $X_1$ and $X_2$, and check if such an $x$ exists in $X_1$ via the above equation. If it does, then we recurse on $X_1$ to find $x$. Otherwise, such an $x$ must exist in $X_2$ (as it does in $X$), and so we recurse on $X_2$.
To make this process efficient in our new model, we pre-build a binary search tree over the elements of $X$, where the internal nodes contain all the query-sets we need. That is, in the root node we have the query-set for $S\cup X$, and in its two children for $S\cup X_1$ respectively $S\cup X_2$.

Using this binary tree, one can simulate the binary search process as described above.
Since what we need to do in a BFS is to (i) find a replacement element $x$ and (ii) mark $x$ as visited (thus effectively ``deactivate'' $x$ in $X$), each time we see $x$, we just need to remove $x$ from the $O(\log{n})$ nodes on a root-to-leaf path, and thus the whole BFS algorithm runs in near-linear time as well.

\paragraph{Batching, Periodic Rebuilding, and Augmenting Sets.}
The above binary search tree is efficient when the common independent set $S$ is static. However, once we find an augmenting path, we need to update $S$. This means that every node in the binary search tree needs to be updated.
If done naively, this would need at least $\Omega(nr)$ time, as there are up to $r$ augmentations, and rebuilding the tree takes $O(n)$ time.
Therefore, we employ a batching approach here.
That is, we do not walk through every node and update them immediately when we see an update to $S$.
Instead, we batch $k$ updates (for $k$ to be decided later) and pay an additional $O(k)$-factor every time we want to do a query in our tree. In other words, at some point, we might want to search for exchanges for a common independent set $S'$ (by doing queries like $(S' \setminus \{u\}) \cup X$ to find edges incident to $u$). Our binary tree might only have an outdated version $S$ (i.e.\ store sets like $S\cup X$). Then the cost of converting $S\cup X$ to $(S'\setminus \{u\}) \cup X$ is $|S\oplus S'| + 1$, which we assert is less than $k$.
When this number exceeds $k$, we rebuild the binary search tree completely using the up-to-date $S'$ instead, in $\tO(n)$ time.

Over the whole run of the algorithm, there are only $O(r \log r)$ updates to our common independent set $S$ (see, e.g., \cite{cunningham1986improved,nguyen2019note}).
Hence, the total running time becomes
\[ \underbrace{\tO(nk\epsilon^{-1})}_{\text{Blocking-flow for $\tO(\epsilon^{-1})$ iterations}} + \underbrace{\tO(\epsilon rn)}_{\text{Remaining $\tO(\epsilon r)$ augmenting paths}} + \underbrace{\tO(nrk^{-1})}_{\text{Rebuilding binary search trees}}, \]
which is $\tO(nr^{2/3})$ for $k = r^{1/3}$ and $\epsilon =r^{-1/3}$.

To achieve the $\tO(n\sqrt{r})$ bound in our dynamic-rank-oracle model, there is one additional observation we need.
By the ``Augmenting Sets'' argument~\cite{chakrabarty2019faster}, for each element $v$ that we want to query our tree, it suffices to consider changes to $S$ that are in the same distance layer as $v$ is (in a single blocking-flow phase).
Since changes to $S$ are uniformly distributed among layers, when the $(s, t)$-distance in $G(S)$ is $d$, we only need to spend an additional $O(\frac{k}{d})$-factor (instead of an $O(k)$-factor) when querying the binary search tree.
This brings our complexity down to
\[ \tO\left(n\epsilon^{-1} + \sum_{d = 1}^{1/\epsilon}\frac{nk}{d}\right) + \tO(\epsilon rn) + \tO(nrk^{-1}), \]
where the first part is a harmonic sum which makes for $\tO(n\epsilon^{-1} + nk)$, and the total running time is $\tO(n\sqrt{r})$ for $k = r^{1/2}$ and $\epsilon = r^{-1/2}$.

\subsection{Matroid Union}

For simplicity of the presentation in this overview, let's assume we are solving the $k$-fold matroid union problem and that $k$---the number of bases\footnote{As an example, consider the problem of finding $k$ disjoint spanning trees of a graph.} we want to find---is constant.
A standard black-box reduction from matroid intersection, combined with our algorithm outlined above, immediately gives us an $\tO(n\sqrt{r})$ bound in the dynamic-rank-oracle model.
Nevertheless, we show how to exploit certain properties of matroid union (specifically, the structure of the exchange graphs~\cite{edmonds1968matroid,cunningham1986improved} resulted from the reduction below) to speed this up to $\tO(n + r\sqrt{r})$, i.e.\ \emph{near-linear time} for sufficiently ``dense''\footnote{We call matroids with $n \gg r$ ``dense'' by analogy to the graphic matroids where $n$ denotes the number of edges and $r$ the number of vertices.} matroids.

Suppose $\cM = (\ground,\cI)$ is the matroid we want to find $k$ disjoint bases for. The standard reduction to matroid intersection is that we create $k$ copies of all elements $u\in \ground$. Then we define two matroids as follows:
\begin{itemize}
\item The first matroid $\cM_1$ says that we only want to use one version of each element. We set $\cM_1 = (U\times \{1,\ldots,k\}, \cI_\text{part})$ to be the partition matroid defined as $S\in \cI_{\text{part}}$ if and only if $|\{(u,i)\in S : i = x\}| \le 1$ for all $x$.
\item The second matroid $\cM_2$ says that for each copy of the ground set $\ground$ we must pick an independent set according to $\cM$. That is set $\cM_2 = (U\times \{1,\ldots,k\},\hat{\cI})$ to be the \emph{disjoint} $k$-fold union of $\cM$, i.e.\ $S\in \hat{\cI}$ if and only if $\{u : (u,i)\in S\}$ is independent in $\cM$ for all $i$.
\end{itemize}
For a set $S$ which can be partitioned into $k$ disjoint independent sets, notice that in the exchange graph, the number of elements \emph{not} in the first layer is bounded by $O(r)$.
This is because every $u\in \ground$ who is not represented in $S$ will be in the first layer $L_1$ of the BFS tree.
As such, we can build the BFS layers starting from the second layer if we can identify all the elements in this second layer.
This can be done by checking for each $y$ \emph{not} in the first layer whether $L_1$ contains an exchange element of $y$ (via computing the rank of $(S \setminus \{y\}) \cup L_1$; no need to do a binary search).
Although binary search is not needed when identifying elements in the second layer, when going backward among layers to find an augmenting path $P$, we still have to find the exact element in the first layer which can be the first element of $P$ since it will decide which augmenting paths remain ``compatible'' later.
This inspires us to maintain two separate binary search trees: one, of size $O(r)$, for finding edges from the second layer and onward, and the other, of size $O(n)$, for finding the first elements of the augmenting paths.
Still, doing a binary search for each element in the first layer results in a total number of $O(r\epsilon^{-1})$ queries to the binary search tree, which is too much.
To reduce the number of queries down to $\tO(r)$, we note that only binary searches which correspond to the actual augmenting paths will succeed, i.e., reach the leaf nodes of the binary search tree.
Since there are at most $O(r/d)$ augmenting paths when the $(s, t)$-distance in $G(S)$ is $d$, we only need to do $O(r/d)$ queries to the binary search tree; other queries can be blocked by first checking if their corresponding exchange elements exist in the first layer.
This results in a running time of $\tO(n + r\sqrt{n})$ (note: $\sqrt{n}$ and not $\sqrt{r}$), which already matches Gabow's algorithm for $k$-disjoint spanning tree~\cite{gabow1988forests}. 

\paragraph{Toward $\tO(n + r\sqrt{r})$ for Matroid Union.}
The bottleneck of the above algorithm is that we need to do binary searches over (and hence rebuild periodically) the tree data structure for the first layer (of size $\Omega(n)$).
If we can reduce the size of this tree down to $O(r)$, then the running time would be $\tO(n + r\sqrt{r})$.
This suggests that we might want to somehow ``sparsify'' the first layer.
Indeed, for a single augmenting path, we only need a basis of the first layer.
As a concrete example, consider the case of a graphic matroid: Given a forest $S$, an edge $e \in S$, and the set of non-tree edges $E \setminus S$, we want to find a ``replacement'' edge $e^\prime$ in $E \setminus S$ for $e$ which ``restores'' the connectivity of $S - e + e^\prime$.
In this case, it suffices to only consider a spanning forest (i.e.\ ``basis'') $B$ of $E \setminus S$, in the sense that such a replacement edge exists in $E \setminus S$ if and only if it exists in this spanning forest $B \subseteq E\setminus S$.

Moreover, note that after each augmentation a single element will be removed from the first layer.
Thus, if we can maintain a decremental basis of the first layer, we can build our binary search tree data structure dynamically on top of this basis and get the desired time bound.

\paragraph{Maintaining a Basis in a Matroid.}
Our data structure for maintaining a basis is inspired by the dynamic minimum spanning tree algorithm of \cite{Frederickson85}, in combination with the sparsification trick of \cite{EppsteinGIN97}. It uses $\tO(n)$ time to initialize, and then $\tO(\sqrt{r})$ dynamic rank queries\footnote{In the application of our matroid union algorithm, there will only be $\tO(r)$ updates, so this is efficient enough for our final $\tO(n+r\sqrt{r})$ algorithm.} per deletion.
It also supports maintaining a min-weight basis.

Let $L \subseteq \ground$ for $|L| = O(n)$ be the first layer in which we want to maintain a dynamic basis.
In the preprocessing stage, we split $L$ into $\sqrt{n}$ blocks $L_1, L_2, \ldots, L_{\sqrt{n}}$ of size roughly $\sqrt{n}$ and compute the basis of $L$ from left to right.
We also build the ``prefix sums'' of these $\sqrt{n}$ blocks so that we can quickly access/query sets of the form $L_1 \cup L_2 \cup \cdots \cup L_k$ for all values of $k$.
When we remove an element $x$ from $L_i$, we first update the prefix sums in $O(\sqrt{n})$ time.
If $x$ is not in the basis we currently maintain, then nothing additional needs to be done.
Otherwise, we have to find the ``first'' replacement element, which is guaranteed to be located in blocks $L_i, \ldots, L_{\sqrt{n}}$.
The block $L_j$ in which the replacement element lies can be identified simply by inspecting the ranks of the prefix sums, and after that, we then go through elements in that block to find the exact element.
Note that blocks after $L_j$ need not be updated, as for them it does not matter what basis we picked among blocks $L_1$ to $L_j$. 
This gives us an $O(\sqrt{n})$-update-time algorithm for maintaining a basis of a matroid.

To get a complexity of $O(\sqrt{r}\log n)$, we show that a similar sparsification structure as that of \cite{EppsteinGIN97} for dynamic graph algorithms also works for arbitrary matroids. The sparsification is a balanced binary tree over the $n$ elements, where in each node we have an instance of our (un-sparsified) underlying data structure to maintain a basis consisting of elements in the subtree rooted at the node. Only elements part of the basis of a node are propagated upwards to the parent node. This means that in each instance of our underlying data structure we work over a ground set of size at most $2r$.
Thus, each update corresponds to at most two updates (a single insertion and deletion) to at most $O(\log n)$ (which is the height of the tree) nodes of the tree, each costing $O(\sqrt{r})$ dynamic rank queries in order to maintain the basis at this node.
This results in the desired time bound.

\section{Preliminaries} \label{sec:prelim}

\paragraph{Notation.}
We use standard set notation.
In addition to that, for two sets $X$ and $Y$, we use $X + Y$ to denote $X \cup Y$ (when $X \cap Y = \emptyset$) and $X - Y$ to denote $X \setminus Y$ (when $Y \subseteq X$).
For an element $v$, $X + v$ and $X - v$ refer to $X + \{v\}$ and $X - \{v\}$, respectively.
Let $X \oplus Y$ denote the symmetric difference of $X$ and $Y$.

\paragraph{Matroid.}
In this paper, we use the standard notion of matroids which is defined as follows.

\begin{definition}
A \emph{matroid} $\cM = (\ground, \cI)$ is defined by a tuple consisting of a finite \emph{ground set} $\ground$ and a non-empty family of \emph{independent sets} $\cI \subseteq 2^\ground$ such that the following properties hold.
\begin{itemize}
  \item \textbf{Downward closure:} If $S \in \cI$, then any subset $S^\prime \subseteq S$ is also in $\cI$.
  \item \textbf{Exchange property:} For any two sets $S_1, S_2 \in \cI$ with $|S_1| < |S_2|$, there exists an $x \in S_2 \setminus S_1$ such that $S_1 + x \in \cI$.
\end{itemize}
\label{def:matroid}
\end{definition}

Let $\ground$ be the ground set of a matroid $\cM$.
For $S \subseteq \ground$, let $\bar{S}$ denote $\ground \setminus S$.
For $X \subseteq \ground$, the \emph{rank} of $X$, denoted by $\rank(X)$, is the size of the largest independent set contained in $X$, i.e., $\rank(X) = \max_{S \in \cI}|X \cap S|$.
The \emph{rank} of a matroid $\cM = (\ground, \cI)$ is the rank of $\ground$.
We let $r \leq n$ denote the rank of the input matroids.
When the input consists of more than one matroid (e.g., in the matroid union problem), let $\rank_i$ denote the rank function of the $\ith$ matroid.
A \emph{basis} of $X$ is an independent set $S \subseteq X$ with $|S| = \rank(X)$.
A basis of $\cM$ is a basis of $\ground$.
The \emph{span} of $X$ contains elements whose addition to $X$ does not increase the rank of it, i.e., $\sspan(X) = \{u \in \ground \mid \rank(X \cup \{u\}) = \rank(X)\}$.

\begin{fact}
  The rank function is submodular.
  That is, $\rank(X) + \rank(Y) \geq \rank(X \cap Y) + \rank(X \cup Y)$ holds for each $X, Y \subseteq \ground$.
\end{fact}

\begin{fact}[see, e.g., {\cite[Lemma 1.3.6]{price2015}}]
  $\rank(A) = \rank(\sspan(A))$ holds for every $A \subseteq \ground$.
  \label{fact:span-does-not-increase-rank}
\end{fact}

\begin{lemma}
  For two sets $X, Y$ and their bases $S_X, S_Y$, it holds that $\rank(S_X + S_Y) = \rank(X + Y)$.
  \label{lemma:basis-rank}
\end{lemma}
\begin{proof}
  Since $X, Y \subseteq \sspan(S_X + S_Y)$, we have $S_X + S_Y \subseteq X + Y \subseteq \sspan(S_X + S_Y)$.
  The lemma then follows from $\rank(S_X + S_Y) = \rank(\sspan(S_X + S_Y))$ using \cref{fact:span-does-not-increase-rank}.
\end{proof}

\paragraph{Exchange Graph.}
Our algorithms for matroid intersection and union will be heavily based on finding augmenting paths in exchange graphs.

\begin{definition}[Exchange Graph]
For two matroids $\cM_1 = (\ground, \cI_1)$ and $\cM_2 = (\ground, \cI_2)$ over the same ground set and an $S \in \cI_1 \cap \cI_2$, the \emph{exchange graph} with respect to $S$ is a directed bipartite graph $G(S) = (\ground \cup \{s, t\}, E_S)$ with $s, t \not\in \ground$ being two distinguished vertices and $E_S = E_1 \cup E_2 \cup E_s \cup E_t$, where
\begin{align*}
  E_1 &= \{(x, y) \mid x \in S, y \not\in S,\;\text{and}\;S - x + y \in \cI_1\}, \\
  E_2 &= \{(y, x) \mid x \in S, y \not\in S,\;\text{and}\;S - x + y \in \cI_2\}, \\
  E_s &= \{(s, x) \mid S + x \in \cI_1 \},\;\text{and} \\
  E_t &= \{(x, t) \mid S + x \in \cI_2 \}.
\end{align*}
\label{def:exchange-graph}
\end{definition}

The \emph{distance layers} of $G(S)$ is the sets $L_1, \ldots, L_{d_{G(S)}(s, t) - 1}$, where $L_\ell$ consists of elements in $\ground$ that are of distance $\ell$ from $s$ in $G(S)$.
Most matroid intersection algorithms including ours are based on augmenting a common independent set with an augmenting path in $G(S)$ until such a path does not exist.
The following lemma certifies the correctness of this approach.

\begin{lemma}[Augmenting Path]
  Let $P$ be a shortest $(s, t)$-path\footnote{In fact, $P$ only needs to be ``chordless''~\cite{quadratic2021}, i.e., without shortcuts. Nonetheless, a shortest $(s, t)$-path suffices for our rank-query algorithms.} of $G(S)$.
  Then, the set $S^\prime := S \oplus (V(P) \setminus \{s, t\})$ is a common independent set with $|S^\prime| = |S| + 1$.
  On the other hand, if $t$ is unreachable from $s$ in $G(S)$, then $S$ is a largest common independent set.
  \label{lemma:augmenting-path}
\end{lemma}

We write $S \oplus P$, where $P$ is an augmenting path in $G(S)$, for the common independent set $S^\prime := S \oplus (V(P) \setminus \{s, t\})$ obtained by augmenting $S$ along $P$.
Let $d_{G(S)}(u, v)$ denote the $(u, v)$-distance in $G(S)$.
When $S$ is clear from context, let $d_t$ denote $d_{G(S)}(s, t)$.
The following lemma states that if $d_{G(S)}(s, t)$ is large, then $S$ is close to being optimal.

\begin{lemma}[\cite{cunningham1986improved}]
  If $S \in \cI_1 \cap \cI_2$ satisfies $d_{G(S)}(s, t) \geq d$, then $|S|$ is at least $\left(1 - O(\frac{1}{d})\right)r$.
  \label{lemma:approx}
\end{lemma}

The following bound on the total length of shortest augmenting paths will be useful for our analysis.

\begin{lemma}[\cite{cunningham1986improved}]
  If we solve matroid intersection by repeatedly finding the shortest augmenting paths, then the sum of the lengths of these augmenting paths is $O(r\log r)$.
  \label{lemma:augmenting-path-lengths}
\end{lemma}

\begin{lemma}[\cite{cunningham1986improved,price2015,chakrabarty2019faster}]
  If we augment along a shortest $(s, t)$-path in $G(S)$ to obtain $S^\prime$, then for each $u \in \ground$, the following hold (let $d := d_{G(S)}$ and $d^\prime := d_{G(S^\prime)})$.
  \begin{enumerate}
    \item If $d(s, u) < d(s, t)$, then $d^\prime(s, u) \geq d(s, u)$. If $d(u, t) < d(s, t)$, then $d^\prime(u, t) \geq d(u, t)$.
    \item If $d(s, u) \geq d(s, t)$, then $d^\prime(s, u) \geq d^\prime(s, t)$. If $d(u, t) \geq d(s, t)$, then $d^\prime(u, t) \geq d^\prime(s, t)$.
  \end{enumerate}
  \label{lemma:monotone}
\end{lemma}

\paragraph{Augmenting Sets.}
The following notion of \emph{augmenting sets}, introduced by \cite{chakrabarty2019faster}, models a collection of ``mutually compatible'' augmenting paths, i.e., paths that can be augmented sequentially without interfering with each other.

\begin{definition}[Augmenting Set~{\cite[Definition 24]{chakrabarty2019faster}}]
  Let $S \in \cI_1 \cap \cI_2$ satisfy $d_{G(S)}(s, t) = d_t$ and let $L_1, L_2, \ldots, L_{d_t - 1}$ be the distance layers of $G(S)$.
  A collection $\Pi := (D_1, D_2, \cdots, D_{d_t - 1})$ is an \emph{augmenting set} in $G(S)$ if
  \begin{enumerate}[noitemsep, label=(\roman*)]
    \item $D_\ell \subseteq L_\ell$ holds for each $1 \leq \ell < d_t$,
    \item $|D_1| = |D_2| = \cdots = |D_{d_t - 1}|$,
    \item $S + D_1 \in \cI_1$,
    \item $S + D_{d_t - 1} \in \cI_2$,
    \item $S - D_\ell + D_{\ell + 1} \in \cI_1$ holds for each \emph{even} $1 \leq \ell < d_t - 1$, and
    \item $S - D_{\ell + 1} + D_{\ell} \in \cI_2$ holds for each \emph{odd} $1 \leq \ell < d_t - 1$.
  \end{enumerate}
  \label{def:augmenting-sets}
\end{definition}

One can think of the concept of augmenting sets as a generalization of augmenting paths.
Indeed, an augmenting path is an augmenting set where $|D_1| = \cdots = |D_{d_t - 1}| = 1$.
The term ``mutually compatible'' augmenting paths is formalized as follows.

\begin{definition}[Consecutive Shortest Paths~{\cite[Definition 28]{chakrabarty2019faster}}]
  A collection of vertex-disjoint shortest $(s, t)$-paths $\cP = (P_1, \ldots, P_k)$ in $G(S)$ is a collection of \emph{consecutive shortest paths} if $P_i$ is a shortest augmenting path in $G(S \oplus P_1 \oplus \cdots \oplus P_{i - 1})$ for each $1 \leq i \leq k$.
\end{definition}

The following structural lemmas of \cite{chakrabarty2019faster} will be useful for us, particularly in deriving \cref{lemma:key} in \cref{sec:matroid-intersection}.

\begin{lemma}[{\cite[Theorem 25]{chakrabarty2019faster}}]
  Let $\Pi := (D_1, \ldots, D_{d_t - 1})$ be an augmenting set in $G(S)$.
  Then, $S^\prime := S \oplus \Pi := S \oplus D_1 \oplus \cdots \oplus D_{d_t - 1}$ is a common independent set.
  \label{lemma:aug-set}
\end{lemma}

For two augmenting sets $\Pi = (D_1, D_2, \ldots, D_{d_t - 1})$ and $\Pi^\prime = (D_1^\prime, D_2^\prime, \ldots, D_{d_t - 1}^\prime)$, we use $\Pi \subseteq \Pi^\prime$ to denote that $D_{\ell} \subseteq D_{\ell}^\prime$ hold for each $1 \leq \ell < d_t$.
In this case, let $\Pi^\prime \setminus \Pi := (D_1^\prime \setminus D_1, \ldots, D_{d_t - 1}^\prime \setminus D_{d_t - 1})$.
We will hereafter abuse notation and let $\Pi$ also denote the set of elements $D_1 \cup \cdots \cup D_{d_t - 1}$ in it.
In particular, $\ground \setminus \Pi$ denotes $\ground \setminus \left(D_1 \cup \cdots \cup D_{d_t - 1}\right)$.

\begin{lemma}[{\cite[Theorem 33]{chakrabarty2019faster}}]
  For two augmenting sets $\Pi \subseteq \Pi^\prime$ in $G(S)$, $\Pi^\prime \setminus \Pi$ is an augmenting set in $G(S \oplus \Pi)$.
  \label{lemma:difference-is-aug-set}
\end{lemma}

\begin{lemma}[{\cite[Theorem 29]{chakrabarty2019faster}}]
  Given a collection of consecutive shortest paths $P_1, \ldots, P_k$ in $G(S)$, where $P_i = (s, a_{i,1}, \ldots, a_{i, d_t - 1}, t)$, the collection $\Pi = (D_1, \ldots, D_{d_t - 1})$, where $D_i = \{a_{1,i}, \ldots, a_{k, i}\}$, is an augmenting set in $G(S)$.
  \label{lemma:aug-set-from-aug-paths}
\end{lemma}

The converse of \cref{lemma:aug-set-from-aug-paths} also holds.

\begin{lemma}[{\cite[Theorem 34]{chakrabarty2019faster}}]
  Given an augmenting set $\Pi$ in $G(S)$, there is a collection consecutive shortest paths $P_1, \ldots, P_k$ in $G(S)$ where $P_i = (s, a_{i,1}, \ldots, a_{i,d_t - 1}, t)$ such that $D_i = \{a_{1,i}, \ldots, a_{k,i}\}$.
  \label{lemma:aug-paths-from-aug-set}
\end{lemma}

\begin{remark}
  Note that \cref{lemma:aug-set-from-aug-paths,lemma:aug-paths-from-aug-set} are not equivalent to the exact statements of {\cite[Theorems 29 and 34]{chakrabarty2019faster}} (in particular, they did not specify how $\Pi$ and $P_i$ are constructed), but our versions are clear from their proof.
\end{remark}

\begin{claim}
  Let $S \in \cI_1 \cap \cI_2$ with $d_{G(S)}(s, t) = d_t$ and $\Pi \subseteq \Pi^\prime$ be two augmenting sets in $G(S)$.
  Let $u \not\in \Pi^\prime$ be an element.
  If $u$ is not on any augmenting path of length $d_t$ in $G(S \oplus \Pi)$, then $u$ is not on any augmenting path of length $d_t$ in $G(S \oplus \Pi^\prime)$ either.
  \label{claim:no-path}
\end{claim}

\begin{proof}
  Let $S^\prime := S \oplus \Pi$ and $S^{\prime\prime} := S \oplus \Pi^\prime$.
  Since $S^{\prime\prime}$ can be obtained by augmenting $S^\prime$ along a series of shortest augmenting paths
  (by \cref{lemma:difference-is-aug-set,lemma:aug-paths-from-aug-set}),
  the claim follows from the fact that the $(s, u)$-distance and $(u, t)$-distance are monotonic (\cref{lemma:monotone}).
\end{proof}

\begin{claim}
  Let $\Pi = (D_1, \ldots, D_{d_t - 1})$ be an augmenting set in $G(S)$ and $P = (s, a_1, \ldots, a_{d_t - 1}, t)$ be an augmenting path in $G(S \oplus \Pi)$.
  Then, $\Pi^\prime = (D_1 + a_1, \ldots, D_{d_t - 1} + a_{d_t - 1})$ is an augmenting set in $G(S)$.
  \label{claim:put-aug-path-in-aug-set}
\end{claim}

\begin{proof}
  This directly follows from
  \cref{lemma:aug-paths-from-aug-set,lemma:aug-set-from-aug-paths}.
\end{proof}

\paragraph{Using Dynamic Oracle.}
In the following sections except for \cref{appendix:dynamic-matroid-intersection-ind} where we discuss independence-query algorithms, all algorithms and data structures will run in the \emph{dynamic-rank-oracle} model (see \cref{def:dyn-oracle}).
In other words, we will simply write ``in $t$ time'' for ``in $t$ time and dynamic rank queries''.
We will use the term \emph{query-sets} to refer to the sets $S_i$ in \cref{def:dyn-oracle}.
In particular, constructing a query-set means building the corresponding set from $S_0 = \emptyset$ with the $\textsc{Insert}(\cdot)$ operation.
Insertion/Deletion of an element into/from a query-set is done via the $\textsc{Insert}$/$\textsc{Delete}$ operations.
Using the $\textsc{Query}$ operation, we assume that we know the ranks of all the query-sets we construct in our algorithms.

\section{Binary Search Tree} \label{sec:bst}

In this section, we give the core data structure of our algorithms which allows us to do binary searches and find free elements (elements $x$ such that $S + x \in \cI$) and exchange pairs (pairs $(x,y)$ such that $S-x+y\in \cI$, corresponding to edges in the exchange graph) efficiently.
We also support updating the common independent set $S$ that the exchange relationship is based upon.
For a matroid $\cM = (\ground, \cI)$, the data structure has the following guarantee ($s, t \not\in U$ denote the two distinguished vertices of the exchange graph as defined in \cref{def:exchange-graph}).

\begin{theorem}
  For any integer $\beta \geq 1$, there exists a data structure that supports the following operations.
  \begin{itemize}
    \item $\textsc{Initialize}(\cM, S, Q_S, X)$: Given $S \in \cI$, the query-set $Q_S$ that corresponds to $S$, and $X \subseteq \bar{S}$ (respectively, $X \subseteq S$ or $X = \{t\}$), initialize the data structure in $\tO(|X|)$ time. The data structure also maintains $S$.
    \item $\textsc{Find}(y)$: Given $y \in S \cup \{s\}$ (respectively, $y \in \bar{S}$),
      \begin{itemize}
        \item if $y \in S$ (respectively, $X \subseteq S$), then return an $x \in X$ such that $S - y + x \in \cI$ (respectively, $S - x + y \in \cI$), or
        \item if $y = s$ (respectively, $X = \{t\}$), then return an $x \in X$ such that $S + x \in \cI$ (respectively, return the only element $x = s$ or $x = t$ in $X$ if $S + y \in \cI$ and $\bot$ otherwise).
      \end{itemize}
      The procedure returns $\bot$ if such an $x$ does not exist.
      The procedure takes $\tO(\beta)$ time if the result is not $\bot$, and $\tO(1)$ time otherwise.
    \item $\textsc{Delete}(x)$: Given $x \in X$, if $x \not\in \{s, t\}$, delete $x$ from $X$ in $O(\log{n})$ time.
    \item $\textsc{Replace}(x, y)$: Given $x \in X$ and $y \not\in X$, replace $x$ in $X$ by $y$ in $O(\log{n})$ time.
    \item $\textsc{Update}(\Delta)$: Update $S$ to $S \oplus (\Delta \setminus \{s, t\})$ in amortized $\tO(\frac{|X| \cdot |\Delta|}{\beta})$ time.
  \end{itemize}
  \label{thm:bst}
\end{theorem}

\begin{remark}
  To make sense of the seemingly complicated input and casework of \cref{thm:bst}, one should focus on the first item of $\textsc{Find}(\cdot)$.
  We will use \cref{thm:bst} to explore the exchange graphs, and thus we need to find an exchange element $x$ of $y$ as in the first case.
  The additional complication is included solely because we also have to deal with edges incident to $s$ or $t$.
  For instance, say $X \subseteq \bar{S}$, then $\textsc{Find}(s)$ finds an edge in $G(S)$ directed from $s$ to $X$.
  This will make our algorithms presented later cleaner (see \cref{alg:blocking-flow} for example).
  \label{remark:bst}
\end{remark}

Sometimes, we will omit the $Q_S$ parameter of $\textsc{Initialize}$, meaning that we explicitly build the query-set $Q_S$ from $S$ in $O(r)$ time before running the actual initialization.
In such cases, $|X|$ will be $\Omega(r)$, and thus this incurs no overhead.

We will later refer to the case of $X \subseteq \bar{S}$ as the \emph{co-circuit binary search tree} and the case of $X \subseteq S$ as the \emph{circuit binary search tree}.
The data structure follows from the binary search algorithm of \cite[Lemma 10]{chakrabarty2019faster}, which is based on the following observation.

\begin{observation}[\cite{chakrabarty2019faster}]
  \label{obs:exchange}
  To find free elements and exchange pairs, we can use the following observations.
\begin{enumerate}[label=(\roman*)]
\item\label{item:free-element}
\emph{Free element}:
  There exists an $x \in X$ such that $S + x \in \cI$ if and only if $\rank(S + X) > |S|$.
\item\label{item:cocircuit-exchange}
\emph{Co-circuit exchange}:
  Given $y\in S$, there exists an $x \in X$ such that $S - y + x \in \cI$ if and only if $\rank(S - y + X) \geq |S|$.
  \item\label{item:circuit-exchange}
\emph{Circuit exchange}:
  Given $y\not\in S$,
  there exists an $x \in X$ such that $S - x + y \in \cI$ if and only if $\rank(S - X + y) = |S-X+y|$.
  \end{enumerate}
\end{observation}

The data structure of \cref{thm:bst} is built upon the following similar data structure whose independent set $S$ is ``static'' in the sense that its update will be specified for each query.
We construct the data structure of \cref{lemma:bst} first, and then use it for \cref{thm:bst} 
later in the section.

\begin{lemma}
  There exists a data structure that supports the following operations.
  \begin{itemize}
    \item $\textsc{Initialize}(\cM, S, Q_S, X)$: Given $S \in \cI$, a query-set $Q_S$ corresponding to $S$, and $X \subseteq \bar{S}$ (respectively, $X \subseteq S$ or $X = \{t\}$), initialize the data structure in $\tO(|X|)$ time.
    \item $\textsc{Find}(y, \Delta)$: Given $\Delta \subseteq V$, let $S^\prime := S \oplus \Delta$. It is guaranteed that $S^\prime \in \cI$. Given $y \in S^\prime \cup \{s\}$ (respectively, $y \in \bar{S^\prime}$),
      \begin{itemize}
        \item if $y \in S^\prime$ (respectively, $X \subseteq S^\prime$), then return an $x \in X$ such that $S^\prime - y + x \in \cI$ (respectively, $S^\prime - x + y \in \cI$), otherwise
        \item if $y = s$ (respectively, $X = \{t\}$), then return an $x \in X$ such that $S^\prime + x \in \cI$ (respectively, return the only element $x = s$ or $x = t$ in $X$ if $S^\prime + y \in \cI$ and $\bot$ otherwise),
      \end{itemize}
      in $\tO(\beta)$ time.
      The procedure returns $\bot$ if such an $x$ does not exist.
    \item $\textsc{Delete}(x)$: Given $x \in X$, delete $x$ from $X$ in $O(\log{n})$ time.
    \item $\textsc{Replace}(x, y)$: Given $x \in X$ and $y \not\in X$, replace $x$ by $y$ in $X$ in $O(\log{n})$ time.
  \end{itemize}  
  \label{lemma:bst}
\end{lemma}

We present the co-circuit version of the data structures as the circuit version is analogous (their difference is essentially stated in the two cases \ref{item:cocircuit-exchange} and \ref{item:circuit-exchange} of \cref{obs:exchange}).
The data structure of \cref{lemma:bst} is a balanced binary tree in which every node $v$ corresponds to a subset $X_v$ of $X$.
The subsets corresponding to nodes at the same level form a disjoint partition of $X$.
There are $|X|$ leaf nodes, each of which corresponds to a single-element subset of $X$.
An internal node $v$ with children $u_1$ and $u_2$ has $X_v = X_{u_1} \sqcup X_{u_2}$.
Each node $v$ is also associated with a query-set $Q_v := S + X_v$, for which we have prepared a dynamic oracle (see \cref{def:dyn-oracle}).

\paragraph{Initialization.}
In the initialization stage, we first compute the query-set of the root node $Q_r := S + X$ from $Q_S$ in $O(|X|)$ time.
As long as the current node $v$ has $|X_v| > 1$, we split $X_v$ into two equally-sized subsets $X_{u_1}, X_{u_2}$, compute $Q_{u_1}, Q_{u_2}$ from $Q_v$, and then recurse on the two newly created nodes $u_1$ and $u_2$.
Computing $Q_{u_1}$ and $Q_{u_2}$ from $Q_v$ takes $O(|X_v|)$ time in total, and thus the overall running time for initialization is $\tO(|X|)$.

\paragraph{Query.}
To find an exchange element of $y \in S^\prime$, we perform a binary search on the tree.
For each node $v$, we can test whether such an element exists in $X_v$ via \cref{obs:exchange}\ref{item:cocircuit-exchange} by computing the query-set $Q^\prime_v := S^\prime - y + X_v$ from $Q_v := S + X_v$ in $1+|\Delta|$ dynamic-oracle queries.
If such an element does not exist for the root $X_r = X$, then we return $\bot$.
Otherwise, for node $v$ initially being $r$, there must exist one of the child nodes $u_i$ of $v$ where such an exchange $x$ exists in $X_{u_i}$.
We then recurse on $u_i$ until we reach a leaf node, at which point we simply return the corresponding element.
Similarly, to find a free element, we compute the rank of $Q^\prime_v := S^\prime + X_v$ instead (see \cref{obs:exchange}\ref{item:free-element}).
Since we need to compute $Q^\prime_v$ for each of the visited nodes, the running time is $O(|\Delta|\log{n})$.

\paragraph{Update.}
For deletion of $x$, we simply walk up from the leaf node corresponding to $x$ to the root node and remove $x$ from each of the $X_v$ and $Q_v$.
This takes time proportional to the depth of the tree, which is $O(\log{n})$.
Replacement of $x$ by $y$ follows similarly from deletion of $x$: instead of simply removing $x$ from $X_v$ and $Q_v$, we add $y$ to them as well.

\begin{remark}
  Note that the above binary search tree is static in the sense that we only deactivate elements from a fixed initial set.
  We can extend this data structure to support a dynamically changing input set $X$ by using a dynamic binary search tree based on partial rebuilding~\cite{Andersson89,Andersson91} instead.
  The amortized time complexity remains the same since rebuilding a subtree takes time proportional to the number of nodes of it.
\end{remark}

\subsection{Periodic Rebuilding} \label{sec:periodic-rebuild}

Here we extend \cref{lemma:bst} to prove \cref{thm:bst}. Recall that the difference between the two data structures is that we need to support a dynamically changing independent set in \cref{thm:bst} (which we will need since $S$ changes after each augmentation in our matroid algorithms).
How we achieve this is to essentially employ a batch-and-rebuild approach to the binary search tree of \cref{lemma:bst}.

\begin{proof}
  We maintain a binary search tree $\cT$ constructed with $\textsc{Initialize}(\cM, S, Q_S, X)$ of \cref{lemma:bst} and a collection of ``batched'' updates $\Delta_{\text{batch}}$ of size at most $\beta$.
  Throughout the updates, we also maintain the query-set corresponding to the current $S$ starting from the given $Q_S$ and the query-set corresponding to $S + X$, which initially can be computed from $Q_S$ in $O(|X|)$ time.
  Each call to $\textsc{Find}(y)$ is delegated to $\cT.\textsc{Find}(y, \Delta_{\text{batch}})$, which runs in time $\tO(|\Delta_{\text{batch}}|) = \tO(\beta)$ time.
  Note that we can test whether the result of $\cT.\textsc{Find}(\cdot)$ will be $\bot$ in $\tO(1)$ time by simply checking if \cref{obs:exchange}\ref{item:cocircuit-exchange} (or \ref{item:free-element} if $y = s$) holds with the query-set corresponding to $S + X$ we maintain.
  
  Each call to $\textsc{Delete}(x)$ and $\textsc{Replace}(x, y)$ translates simply to $\cT.\textsc{Delete}(x)$ and $\cT.\textsc{Replace}(x, y)$.
  For an update to $S$ with $\Delta$, we set $\Delta_{\text{batch}} \gets \Delta_{\text{batch}} \cup \Delta$ and update $S$ and the query-sets accordingly.
  If the size of $\Delta_{\text{batch}}$ exceeds $\beta$, then we rebuild the binary search tree with the input common independent set being the up-to-date $S$ we maintain.
  Note that we will pass query-set $Q_S$ to $\textsc{Initialize}$ to not pay the extra $O(r)$ factor.
  Finally, since the binary search tree is now up-to-date, we set $\Delta_{\text{batch}}$ to be $\emptyset$.
  The rebuilding takes $\tO(|X|)$ time and is amortized to $\tO(\frac{|X| \cdot |\Delta|}{\beta})$ per update operation with $|\Delta|$ changes.
\end{proof}

\section{Matroid Intersection} \label{sec:matroid-intersection}

In this section, we present a matroid intersection algorithm in the dynamic-rank-oracle model that matches the state-of-the-art algorithm~\cite{chakrabarty2019faster} in the traditional model.

\begin{restatable}{theorem}{matroidintersection}
  For two matroids $\cM_1 = (\ground, \cI_1)$ and $\cM_2 = (\ground, \cI_2)$, it takes $\tO(n\sqrt{r})$ time to obtain the largest $S \in \cI_1 \cap \cI_2$ in the dynamic-rank-oracle model.
  \label{thm:dynamic-matroid-intersection-rank-main}
\end{restatable}

The algorithm follows the blocking-flow framework of \cite{chakrabarty2019faster} similar to the Hopcroft-Karp algorithm for bipartite matching~\cite{HopcroftK73}, which goes as follows.
Initially, they start with $S = \emptyset$.

\begin{enumerate}
  \item\label{item:step1} First, they obtain a common independent set that is of size at least $(1 - \epsilon)r$ by eliminating all augmenting paths of length $O(1/\epsilon)$. In each of the $O(1/\epsilon)$ iterations, they first compute the distance layers of $G(S)$ along which they find a maximal set of compatible shortest augmenting paths using an approach similar to a depth-first-search from $s$.
  Augmenting paths are searched in a depth-first-search manner.
  Whenever an element has no out-edge with respect to the current common independent set to the next layer, they argue that it can be safely removed as it will not be on a shortest augmenting path anymore in this iteration.
  Augmenting along these augmenting paths increases the $(s, t)$-distance of $G(S)$ by at least one.
  \item\label{item:step2} With the current solution which is only $\epsilon$ fraction away from being optimal, they find the remaining $O(\epsilon r)$ augmenting paths one at a time.
\end{enumerate}

A proper choice of $\epsilon$ (in this case it is $\epsilon = 1/\sqrt{r}$) that balances the cost between the two steps results in their algorithm.

\subsection{Building Distance Layers}

Building distance layers and finding a single augmenting path
in Step~\ref{item:step2}
is immediate by replacing binary searches in \cite[Algorithm 4]{chakrabarty2019faster} with the binary search trees of \cref{thm:bst}.

\begin{lemma}
  It takes $\tO(n)$ time to compute the $(s, u)$-distance for each $u \in \ground$ and find the shortest $(s, t)$-path in $G(S)$ or determine that $t$ is unreachable from $s$.
  \label{lemma:bfs}
\end{lemma}

\begin{proof}
  First, we build two binary search trees via \cref{thm:bst} with $\beta = 1$, a circuit binary search tree $\T_1 := \textsc{Initialize}(\cM_1, S, X_1)$ where $X_1 = S$ for the first matroid and a co-circuit binary search tree $\T_2 := \textsc{Initialize}(\cM_2, S, X_2)$ where $X_2 = \bar{S}$ for the second matroid. Initializing these takes $\tO(n)$ time.
  These two binary search trees allow us to explore the exchange graph efficiently.
  
  Then we run the usual BFS algorithm from the source $s$ (or equivalently, all $u \in \bar{S}$ with $S + u \in \cI_1$).
  For each visited element $u$, if $u \in S$, then we repeatedly find $x \in X_2$ such that $S - u + x \in \cI_2$ using $\T_2.\textsc{Find}(u)$, mark $x$ as visited, and remove $x$ from $X_2$ via $\T_2.\textsc{Delete}(x)$ (until $\bot$ is returned).
  Similarly, for $u \in \bar{S}$, we find $x \in X_1$ with $S - x + u \in \cI_1$, mark $x$ as visited, and remove $x$ from $X_1$ using $\T_1$.
  This explores all the unvisited out-neighbors of $u$ in $G(S)$.
  Since each element will be visited at most once, the total running time is $\tO(n)$.
\end{proof}

\subsection{Blocking Flow}

In this section, we prove the following lemma regarding a single phase of blocking-flow computation.

\begin{lemma}
  Given an $S \in \cI_1 \cap \cI_2$ with $d_{G(S)}(s, t) = d$, it takes $\tO\left(n + \frac{n\sqrt{r}}{d} + \frac{(|S^\prime| - |S|) \cdot nd}{\sqrt{r}}\right)$ time to obtain an $S^\prime \in \cI_1 \cap \cI_2$ with $d_{G(S^\prime)}(s, t) > d_{G(S)}(s, t)$.
  \label{lemma:blocking-flow}
\end{lemma}

Before proceeding to prove \cref{lemma:blocking-flow}, we first use it to finish our matroid intersection algorithm.
Like Hopcroft-Karp bipartite matching algorithm \cite{HopcroftK73} and the matroid intersection algorithm of \cite{chakrabarty2019faster}, we run several iterations of blocking-flow, and then keep augmenting until we get the optimal solution.

\begin{proof}[Proof of \cref{thm:dynamic-matroid-intersection-rank-main}]
  Starting from an empty set $S = \emptyset$, we run the blocking-flow algorithm until $d_{G(S)}(s,t) \geq \sqrt{r}$.
  This, by \cref{lemma:blocking-flow}, takes
  \begin{equation}
    \tO\left(n\sqrt{r}\right) + \tO\left(\sum_{d = 1}^{\sqrt{r}}\frac{n\sqrt{r}}{d}\right) + \tO\left(\frac{n}{\sqrt{r}}\cdot\left(\sum_{d = 1}^{\sqrt{r}}{d \cdot (|S_d| - |S_{d - 1}|)}\right)\right)
    \label{eq:timebound}
  \end{equation}
  time, where $S_d$ is the size of the $S$ we get after augmenting along paths of length $d$.
  Observe that $\sum_{d = 1}^{\sqrt{r}}{d \cdot (|S_d| - |S_{d - 1}|)}$ is the sum of lengths of the augmenting paths that we use, and thus the third term in \cref{eq:timebound} is $\tO(\frac{n}{\sqrt{r}} \cdot r) = \tO(n\sqrt{r})$ by \cref{lemma:augmenting-path-lengths}.
  The second term also sums up to $\tO(n\sqrt{r})$ (by a harmonic sum), and therefore the total running time of the blocking-flow phases is $\tO(n\sqrt{r})$.
  The current common independent set $S$ has size at least $r - O(\sqrt{r})$ by \cref{lemma:approx}, and thus finding the remaining $O(\sqrt{r})$ augmenting paths one at a time takes a total running time of $\tO(n\sqrt{r})$ via \cref{lemma:bfs}.
  This concludes the proof of \cref{thm:dynamic-matroid-intersection-rank-main}.
\end{proof}

The rest of the section is to prove \cref{lemma:blocking-flow}.
Our blocking-flow algorithm is a slight modification to \cite[Algorithm 5]{chakrabarty2019faster}, as shown in \cref{alg:blocking-flow}.
It takes advantage of the data structure of \cref{thm:bst} to explore an out-edge from the current element $a_{\ell}$ to $A_{\ell + 1}$---the set of ``alive'' elements in the next layers---while (approximately) keeping track of the current common independent set $S$.
An element $u$ is ``alive'' if it has not been included in the augmenting set $\Pi := (D_1, \ldots, D_{d_t - 1})$ yet, nor has the algorithm determines that there cannot be any shortest augmenting path through $u$.

\begin{algorithm}[!ht]
  \SetEndCharOfAlgoLine{}
  \SetKwInput{KwData}{Input}
  \SetKwInput{KwResult}{Output}
  
  \caption{Blocking flow}
  \label{alg:blocking-flow}
  \KwData{$S \in \cI_1 \cap \cI_2$}
  \KwResult{$S^\prime \in \cI_1 \cap \cI_2$ with $d_{G(S^\prime)}(s, t) > d_{G(S)}(s, t)$}
  
  Build the distance layers $L_1, \ldots, L_{d_t - 1}$ of $G(S)$ with \cref{lemma:bfs}\;
  $L_0 \gets \{s\}$ and $L_{d_t} \gets \{t\}$\;
  $A_\ell \gets L_{\ell}$ for each $0 \leq \ell \leq d_t$\;
  $\cT_{\ell} \gets \textsc{Initialize}(\cM_{1}, S, Q_S, L_{\ell})$ for each \emph{odd} $1 \leq \ell \leq d_t$ by \cref{thm:bst} with $\beta = \sqrt{r} / d$\;
  $\cT_{\ell} \gets \textsc{Initialize}(\cM_{2}, S, Q_S, L_{\ell})$ for each \emph{even} $1 \leq \ell \leq d_t$ by \cref{thm:bst} with $\beta = \sqrt{r} / d$\;
  $\ell \gets 0$, $a_0 \gets s$, and $D_{\ell} \gets \emptyset$ for each $1 \leq \ell < d_t$\;
  \While{$\ell \geq 0$} {
    \If{$\ell < d_t$} {
      \lIf{$A_{\ell} = \emptyset$} {
        \textbf{break}
      }
      $a_{\ell + 1} \gets \cT_{\ell + 1}.\textsc{Find}(a_\ell)$\;\label{line:search}
      \If{$a_{\ell + 1} = \bot$} {
        $\cT_{\ell}.\textsc{Delete}(a_{\ell})$\;
        $A_\ell \gets A_\ell - a_\ell$ and $\ell \gets \ell - 1$\;
      }
      \Else {
        $\ell \gets \ell + 1$
      }
    }
    \Else {
        \tcp{Found augmenting path $a_1, a_2, \ldots a_\ell$}
      \For{$i \in \{1, 2, \ldots, d_t - 1\}$} {
          $D_{i} \gets D_{i} + a_{i}$ and $A_{i} \gets A_{i} - a_{i}$\;
          $\cT_{i}.\textsc{Delete}(a_i)$ and $\cT_{i}.\textsc{Update}(\{a_{i - 1}, a_{i}\})$\;\label{line:update}
      }
      $\ell \gets 0$\;
    }
  }
  \textbf{return} $S^\prime := S \oplus \Pi$, where $\Pi := (D_1, D_2, \ldots, D_{d_t - 1})$\;
\end{algorithm}

We emphasize that the difference between \cref{alg:blocking-flow} and \cite[Algorithm 5]{chakrabarty2019faster} is exactly in the replacement of binary searches with the data structure of \cref{thm:bst}.
Note that indeed by the specification stated in \cref{thm:bst}, the binary search trees let us explore edges in the exchange graph (see \cref{remark:bst}).
As a result, our proof will focus on showing that such a replacement does not affect the correctness.
For this, we need the concept of \emph{augmenting sets} (see \cref{def:augmenting-sets}) which characterizes a collection of ``mutually compatible'' augmenting paths---i.e.\ a ``blocking flow''.
The structural results in \cref{sec:prelim} culminate in the following lemma that is key to the correctness of our algorithm.
It models when we can safely ``remove'' an element since there will be no augmenting path through it in the future.
This is in particular required for us (as opposed to the simpler argument used in \cite{chakrabarty2019faster}) because the set $S$ is not fully updated after each augmentation (at least in the binary search trees that we use to explore the exchange graphs).

\begin{lemma}
  Let $\Pi \subseteq \Pi^\prime$ be augmenting sets in $G(S)$ with distance layers $L_1, \ldots, L_{d_t - 1}$ where $d_{G(S)}(s, t) = d_t$.
  For $x \in L_\ell$, if there is no $y \in L_{\ell + 1}$ such that
  \begin{equation}
    (S \oplus D_{\ell} \oplus D_{\ell + 1}) \oplus \{x, y\} \in \cI,\;\text{where}\;\cI := \begin{cases}\cI_1, & \text{if $\ell$ is even} \\ \cI_2, & \text{if $\ell$ is odd}\end{cases},
    \label{eq:condition}
  \end{equation}
  then there is no augmenting path of length $d_t$ through $x$ in $G(S \oplus \Pi^\prime)$.
  \label{lemma:key}
\end{lemma}

\begin{proof}
  We claim that there is no augmenting path of length $d_t$ through $x$ in $G(S \oplus \Pi)$: If there is such a $P$, then we can put $P$ into $\Pi$ and get an augmenting set $\tilde{\Pi} := (\tilde{D}_1, \ldots, \tilde{D}_{d_t - 1})$ by \cref{claim:put-aug-path-in-aug-set}.
  By definition of the augmenting set, this means that there is such a $y \in \tilde{D}_{k + 1} \setminus D_{k + 1}$ satisfying \eqref{eq:condition}, a contradiction to our assumption.
  The lemma now follows from \cref{claim:no-path}. 
\end{proof}

We are now ready to prove \Cref{lemma:blocking-flow}.

\begin{proof}[Proof of \cref{lemma:blocking-flow}]
  First, We analyze the running time of \cref{alg:blocking-flow}.
  Similar to \cite[Lemma 15]{chakrabarty2019faster}, in each iteration, we use $\cT_\ell.\textsc{Find}(\cdot)$ to find an out-edge of $a_{\ell}$, taking $\tO(\beta) = \tO(\sqrt{r}/d)$ time by \cref{thm:bst}.
  In each iteration, we either increase $\ell$ and extend the current path by a new element, decrease $\ell$ and remove one element, or find an $(s, t)$-path (then remove everything in it), and each element can participate in each of the event at most once.
  Thus, there are only $O(n)$ iterations, and the total cost of $\cT_{\ell}.\textsc{Find}(\cdot)$ is consequently $\tO(\frac{n\sqrt{r}}{d})$ by our choice of $\beta$.
  For each of the augmenting path, $\cT_{\ell}.\textsc{Update}(\cdot)$ takes $\tO\left(\frac{|L_\ell| \cdot d}{\sqrt{r}}\right)$ time, contributing to a total running time of $\tO\left(\frac{(|S^\prime| - |S|) \cdot nd}{\sqrt{r}}\right)$ since $L_\ell$'s are disjoint.
  
  We then argue the correctness of the algorithm.
  Observe that at any point in time, $\cT_\ell$ is a data structure capable of finding a replacement element with respect to the independent set $S \oplus D_{\ell - 1} \oplus D_\ell$, due to the updates that we gave it.
  This means that the collection $\Pi := (D_1, D_2, \ldots, D_{d_t - 1})$ remains an augmenting set in $G(S)$ because $S \oplus (D_{\ell - 1} + a_{\ell - 1}) \oplus (D_{\ell} + a_\ell)$ is independent for each $\ell$ whenever a path is found.
  As a result, when the algorithm terminates, $S^\prime := S \oplus \Pi$ is indeed a common independent set as guaranteed by \cref{lemma:aug-set}.
  
  It remains to show that $d_{G(S^\prime)}(s, t) > d_{G(S)}(s, t)$ by arguing that for each $a_{\ell}$ not in $\Pi$ but removed from $A_{\ell}$ at time $t$, there is no shortest augmenting path in $G(S^\prime)$ that passes through $a_{\ell}$.
  This is a direct consequence of \cref{lemma:key} since $\Pi^{(t)}$, the augmenting set obtained at time $t$, is contained in $\Pi$.
  The fact that $\cT_{\ell + 1}^{(t)}.\textsc{Find}(a_{\ell})$ returns nothing (equivalently, \cref{eq:condition} is not satisfied) shows that $a_{\ell}$ is not on any shortest augmenting path in $G(S^\prime)$ since the set $X$ maintained in $\cT_{\ell + 1}^{(t)}$ (see \cref{thm:bst}) is $A_{\ell + 1}$ at all time.
  We remark that $x$ might have an out-edge (with respect to $S \oplus D_{\ell}^{(t)} \oplus D_{\ell}^{(t + 1)}$) to a removed element with distance $\ell + 1$ from $s$ (not in $A_{\ell + 1}$), but such an element, by induction, is not on any augmenting path either.
\end{proof}

\section{Dynamically Maintaining a Basis of a Matroid} \label{sec:decremental-basis}

In this section, we construct a data structure that allows us to maintain a basis of a matroid in a decremental set.
The data structure is used for obtaining an $\tO_k(n + r\sqrt{r})$ running time for matroid union, but it may be of independent interest as well.
Specifically, our data structure has the following guarantees.

\begin{theorem}
  For a (weighted) matroid $\cM = (U, \cI)$, there exists a data structure supporting the following operations.
  \begin{itemize}
    \item $\textsc{Initialize}(X)$: Given a set $X \subseteq U$, initialize the data structure and return a (min-weight) basis $S$ of $X$ in $\tO(n)$ time.
    \item $\textsc{Delete}(x)$: Given $x \in X$, remove $x$ from $X$ and return a new (min-weight) basis of $X$ in $\tO(\sqrt{r})$ time. Specifically, the new basis will contain at most one element (the replacement element of $x$) not in the old basis, and this procedure returns such an element if any.
  \end{itemize}
  \label{thm:decremental-basis}
\end{theorem}

Our data structure for \cref{thm:decremental-basis} will consist of two parts.
The first part, introduced in \cref{sec:baseline}, is a baseline, unsparsified data structure that supports the $\textsc{Delete}$ operation in $\tO(\sqrt{n})$ time, and the second one is a sparsification structure which brings the complexity down to $\tO(\sqrt{r})$, as presented in \cref{sec:sparsification}.

As hinted by the statement of \cref{thm:decremental-basis}, to make things simpler, we will assign an arbitrary but unique weight $w(x)$ to each $x \in X$.
Now, instead of maintaining an arbitrary basis of $X$, we maintain the \emph{min-weight} basis instead.
The min-weight basis is well-known to be unique (as long as the weights are) and can be obtained greedily as shown in \cref{alg:greedy} (see, e.g.,~\cite{edmonds1971}).

\begin{algorithm}[!ht]
  \SetEndCharOfAlgoLine{}
  \SetKwInput{KwData}{Input}
  \SetKwInput{KwResult}{Output}
  
  \caption{Greedy algorithm for computing the min-weight basis}
  \label{alg:greedy}
  \KwData{A set $X \subseteq U$ of size $k$}
  \KwResult{The min-weight basis $S$ of $X$}
  Order $X = (x_1, x_2, \ldots, x_{k})$ so that $w(x_1) < w(x_2) < \cdots < w(x_{k})$\;
  $S \gets \emptyset$\;
  \For{$i \in [1, k]$} {
    \If{$\rank(S + x_i) > \rank(S)$} {\label{line:check}
      $S \gets S + x_i$\;
    }
  }
  \textbf{return} $S$\;
\end{algorithm}

Moreover, suppose we remove $x \in S$ from the set $X$. Then the new min-weight basis is either
(i) $S - x + y$ where $y$ is the minimum weight element in $X - x$ that makes $S - x + y$ independent or
(ii) simply $S - x$ if such a $y$ does not exist.
In case (i), $y$ is called the \emph{replacement} element of $x$.
Note that $w(y) > w(x)$ must hold.

It is useful to note that the $S$ in Line~\ref{line:check} of \cref{alg:greedy} is interchangeable with $X_{i - 1} = \{x_1, \ldots, x_{i - 1}\}$, since
$\sspan(X_{i - 1}) = \sspan(S \cap X_{i - 1})$, so the sets $X_{i-1}$ and $S\cap X_{i-1}$ have the same rank.
In other words, in each iteration $i$, we can imagine that \cref{alg:greedy} has chosen every element before $x_i$.

\begin{observation}
  In \cref{alg:greedy}, $x_i \in S$ if and only if $\rank(X_{i}) > \rank(X_{i - 1})$.
  \label{obs:greedy}
\end{observation}

\subsection{Baseline Data Structure} \label{sec:baseline}

Our baseline data structure supports the operations of \cref{thm:decremental-basis}, except in time $\tO(\sqrt{k})$ where $k = |X|$ instead of $\tO(\sqrt{r})$.

\begin{lemma}
  For a weighted matroid $\cM = (U, \cI)$, there exists a data structure supporting the following operations.
  \begin{itemize}
    \item $\textsc{Initialize}(X)$: Given a set $X \subseteq U$ with $|X| = k$, initialize the data structure and return the min-weight basis $S$ of $X$ in $\tO(k)$ time.
    \item $\textsc{Delete}(x)$: Given $x \in X$, remove $x$ from $X$ and return the new min-weight basis of $X$ in $\tO(\sqrt{k})$ time. Specifically, the new basis will contain at most one element (the replacement element of $x$) not in the old basis, and this procedure returns such an element if any.
    \item $\textsc{Insert}(x)$: Given $x \not\in X$, add $x$ to $X$. It's guaranteed that $x$ is not in the min-weight basis of the new $X$ and the size of $X$ does not exceed $2k$.
  \end{itemize}
  \label{lemma:decremental-basis_baseline}
\end{lemma}

\paragraph{Initialization.}
In the initialization stage, we order $X$ by the weights and split the sequence into $\sqrt{k}$ blocks $X_1, X_2, \ldots, X_{\sqrt{k}}$ from left to right, where each block has roughly the same size $O(\sqrt{k})$.
That is, $X_1$ contains the $\sqrt{k}$ elements with the smallest weights while $X_{\sqrt{k}}$ contains elements with the largest weights.
We also compute the basis $S$ of $X$ from left to right as in \cref{alg:greedy} together with $\sqrt{k}$ query-sets $Q_1, Q_2, \ldots, Q_{\sqrt{k}}$, where $Q_j = \bigcup_{i = 1}^{j}X_i$ is the union of the first $j$ blocks.
This takes $\tO(k)$ time in total.

\paragraph{Deletion.}
For each deletion of $x$ located in the block $X_i$, we first update the query-sets $Q_i, \ldots, Q_{\sqrt{k}}$ by removing $x$ from them.
Let $Q_1^\prime, \ldots, Q_{\sqrt{k}}^\prime$ denote the old query-sets before removing $x$.
If $x$ is not in the basis $S$ we currently maintain, then $S$ remains the min-weight basis of the new $X$ and nothing further needs to be done.
Otherwise, we would like to find the min-weight replacement element $y$ of~$x$.
We know that such a $y$, if it exists, can only be located in blocks $X_i, X_{i + 1}, \ldots, X_{\sqrt{k}}$.
As such, we find the first $j \geq i$ with $\rank(Q_j) = \rank(Q_j^\prime)$ and recompute the portion of $S$ inside $X_j$.
This can be done by running \cref{alg:greedy} with the initial set $S$ being $Q_{i - 1}$, the union of the first $i - 1$ blocks (see \cref{obs:greedy}).
Thus, the deletion takes $\tO(\sqrt{k})$ time.

\paragraph{Insertion.}
For insertion of $x$, we simply add $x$ to a block where it belongs (according to $w(x)$) and then update $Q_i$'s appropriately.
This takes $\tO(\sqrt{k})$ as well.

\paragraph{Rebalancing.}
To maintain an update time of $\tO(\sqrt{k})$, whenever the size of a block $X_i$ grows larger than $2\sqrt{k}$, we split it into two blocks and recompute $Q_i$ and $Q_{i+1}$.
Similarly, to avoid having too many blocks, whenever the size of a block $X_i$ goes below $\sqrt{k} / 2$, we merge it with an adjacent block and remove $Q_i$.
Each of the above operations takes $\tO(\sqrt{k})$ time, which is subsumed by the cost of an update.
\\

\noindent
We have shown how to implement each operation of \cref{lemma:decremental-basis_baseline} in its desired running time, and the correctness of the data structure is manifest as we always follow the greedy basis algorithm (\cref{alg:greedy}).

\subsection{Sparsification} \label{sec:sparsification}

In this section, we prove \cref{thm:decremental-basis} by ``sparsifying'' the input set of the data structure for \cref{lemma:decremental-basis_baseline} in a recursive manner, similar to what \cite{EppsteinGIN97} did to improve \cite{Frederickson85}'s $O(\sqrt{|E|})$ dynamic MST algorithm to $O(\sqrt{|V|})$.
The following claim asserts that such sparsification is valid.

\begin{claim}
  Let $S_X$ and $S_Y$ be the min-weight basis of $X$ and $Y$, respectively, where $w(x) < w(y)$ holds for each $x \in X$ and $y \in Y$.
  Then, the min-weight basis of $S_X + S_Y$ is also the min-weight basis of $X + Y$.
  \label{claim:sparsification}
\end{claim}

\begin{proof}
  Consider running the greedy \cref{alg:greedy} on the set $X + Y$ to obtain the min-weight basis $S$ of it.
  Clearly, we have $S_X \subseteq S$ since $X$ contains the elements of smaller weights (in fact $S \cap X = S_X$).
  Assume for contradiction that $S \cap Y \not\subseteq S_Y$, i.e., there exists a $y^{*} \in S \cap Y$ which does not belong to $S_Y$.
  Then, it must be the case that there exists a $y \in S_Y$ with $w(y) > w(y^{*})$, as otherwise (i.e., $y^{*}$ is ordered after everything in $S_Y$) by \cref{lemma:basis-rank} the greedy algorithm stops before seeing $y^{*}$.
  We claim that the greedy algorithm on $Y$ chooses $y^{*}$ before all such $y$'s, thereby contradicting the fact that $S_Y$ is the min-weight basis of $Y$.
  This is true by the diminishing returns property\footnote{The diminishing returns property of submodular functions states that $f(z+X)-f(X) \geq f(z+Y)-f(Y)$ holds for each $Y \subseteq X \subseteq \ground$ and $z \not\in X$.} of the rank function: Let $Y^{*}$ be elements in $Y$ with weights smaller than $w(y^{*})$.
  Since $y^{*} \in S$, it follows that $\rank(X + Y^{*} + y^{*}) > \rank(X + Y^{*})$, implying $\rank(Y^{*} + y^{*}) > \rank(Y^{*})$ and the greedy algorithm run on $Y$ picks $y^{*}$.
\end{proof}

We are now ready to present our sparsification data structure.

\begin{proof}[Proof of \cref{thm:decremental-basis}]
  Our data structure is a balanced binary tree where the leaf nodes correspond to elements in $X$ and each internal node corresponds to the set consisting of elements in leaf nodes of this subtree.
  We will abuse notation and use a node $v$ to also refer to the elements contained in the subtree rooted at $v$.
  
  We first build the binary tree top-down, starting with the root node containing $X$ and recursively splitting the current set into two subsets of roughly the same size and recursing on them.\footnote{Note that unlike in \cref{sec:bst}, we are not building query-sets here.}
  We then build the min-weight basis of each node in a bottom-up manner, starting from the leaves.
  For each node $v$ with children $u_1$ and $u_2$, we initialize the data structure $\cD_v$ for \cref{lemma:decremental-basis_baseline} with input set $S_{u_1} + S_{u_2}$, the min-weight basis of $u_1$ and $u_2$ which are obtained from $\cD_{u_1}$ and $\cD_{u_2}$.
  By \cref{claim:sparsification}, the basis $\cD_v$ maintains is the min-weight basis of $v$.
  Thus, by induction, the basis maintained in the root node is indeed the min-weight basis of the whole set $X$.
  The data structure for \cref{lemma:decremental-basis_baseline} takes time near-linear in the size of the input set to construct, and since the sparsified input is a subset of elements in the subtree, the initialization takes time near-linear in the sum of sizes of the subtrees, which is $\tO(n)$ (indeed, every element occurs in at most $\log n$ nodes).
  
  To delete an element $x \in X$, we first identify the leaf node $v_x$ of the binary tree which corresponds to $x$.
  Going upward, for each ancestor $p$ of $v_x$, we delete $x$ from $\cD_p$.
  If we find a replacement element $y$ for $x$, we insert $y$ into $\cD_{q}$, where $q$ is $p$'s parent, before proceeding to $q$ ($y$ is not in $\cD_p$ so such an insertion is valid by \cref{claim:sparsification}).
  Since $x$ will be removed from $\cD_q$ shortly, the input set of $\cD_q$ remains the union of the min-weight bases of $q$'s children.
  This takes $\tO(\sqrt{r})$ time since $\cD_p$ is of size $O(r)$.
  Inductively, since the min-weight bases of the child nodes are updated, by \cref{claim:sparsification}, the min-weight basis of each of the affect nodes (hence the min-weight basis of $X$) is correctly maintained.
\end{proof}

\section{Matroid Union} \label{sec:matroid-union}

In this section, we present our improved algorithm for matroid union.
Our main focus of this algorithm is on optimizing the $O(n\sqrt{r})$ term to $O(r\sqrt{r})$.
Thus, for simplicity of presentation, we will treat $k$ as a constant (the dependence on $k$ will be a small polynomial) and express our bounds using the $O_k(\cdot)$ and $\tO_k(\cdot)$ notation.

\begin{theorem}
  In the dynamic-rank-oracle model, given $k$ matroids $\cM_i = (U_i, \cI_i)$ for $1 \leq i \leq k$, it takes $\tO_k(n + r\sqrt{r})$ time to find a basis $S \subseteq U_1 \cup \cdots \cup U_k$ of $\cM = \cM_1 \vee \cdots \vee \cM_k$ together with a partition $S_1, \ldots, S_k$ of $S$ in which $S_i \in \cI_i$ for each $1 \leq i \leq k$.
  \label{thm:dynamic-matroid-union}
\end{theorem}

In \cref{appendix:matroid-union-fold}, we present an optimized (for the parameter $k$) version of the above algorithm which solves the important special case when all the $k$ matroids are the same---i.e.\ \emph{$k$-fold matroid union}---with applications in matroid packing problems. For example, the problem of finding $k$ disjoint spanning trees in a graph falls under this special case. In particular, in \cref{appendix:matroid-union-fold}, we obtain the following \cref{thm:dynamic-matroid-union-fold}, and we discuss some immediate consequences for the \emph{matroid packing, matroid covering}, and \emph{$k$-disjoint spanning trees problems} in \cref{sec:packing,sec:spanning}.

\begin{restatable}{theorem}{matroidunionfold}
  In the dynamic-rank-oracle model, given a matroid $\cM = (U, \cI)$ and an integer $k$, it takes $\tO(n + kr\sqrt{\min(n, kr)} + k\min(n, kr))$ time to find the largest $S \subseteq U$ and a partition $S_1, \ldots, S_k$ of $S$ in which $S_i \in \cI$ for each $1 \leq i \leq k$.
  \label{thm:dynamic-matroid-union-fold}
\end{restatable}

The rest of this section will focus on the proof of 
\cref{thm:dynamic-matroid-union} (again, where the number of matroids $k$ is treated as a constant).
Our algorithm is based on the matroid intersection algorithm in \cref{sec:matroid-intersection}, in which we identify and optimize several components that lead to the improved time bound.

\subsection{Reduction to Matroid Intersection} \label{sec:reduction}

For completeness, we provide a standard reduction from matroid union to matroid intersection.
For an in-depth discussion, see \cite[Chapter 42]{schrijver2003}.
Let $\cM_i = (U_i, \cI_i)$ be the given $k$ matroids and $U = U_1 \cup \cdots \cup U_k$ be the ground set of the matroid union $\cM = \cM_1 \vee \cdots \vee \cM_k$.
We first relabel each element in the matroids with an identifier of its matroid, resulting in $\hat{\cM}_i = (\hat{U}_i, \cI_i)$, where $\hat{U}_i = \{(u, i) \mid u \in U_i\}$.
Let $\hat{\cM} = (\hat{U}, \hat{\cI}) = \hat{\cM}_1 \vee \cdots \vee \hat{\cM}_k$ be over the ground set $\hat{U} = \hat{U}_1 \sqcup \cdots \sqcup \hat{U}_k$.

In other words, in $\hat{\cM}$, we duplicate each element that is shared among multiple matroids into copies that are considered different, effectively making the ground sets of the $k$ matroids disjoint.
After this modification, an independent set in $\hat{\cM}$ is now simply the union of $k$ independent sets, one from each matroid.
However, that might not be what we want since these independent sets may overlap, i.e., contain copies that correspond to the same element. 
We therefore intersection $\hat{\cM}$ with a partition matroid $\cM_{\text{part}} = (\hat{U}, \cI_{\text{part}})$ given by
\[ \cI_{\text{part}} = \{S \subseteq \hat{U} \mid \left|S \cap \{(u, i)\;\text{for $1 \leq i \leq k$} \mid u \in U_i\}\right| \leq 1\;\text{holds for each $u \in U$} \} \]
to restrict different copies of the same element to be chosen at most once.
The matroid union problem is thus reducible to the matroid intersection problem in the sense that the intersection of $\hat{\cM}$ and $\cM_{\text{part}}$ maps exactly to the independent sets of the matroid union $\cM$.

Notation-wise, given the above mapping between the two worlds, whenever we write $S \in \cI_{\text{part}} \cap \hat{\cI}$, a subset set of $\hat{U}$, we will equivalently regard $S$ as a subset of $U$ with an implicit partition $S_1, \ldots, S_k$ where $S_i \in \cI_i$.

\subsection{Specialized Matroid Intersection Algorithm}

Given the reduction, to prove \cref{thm:dynamic-matroid-union}, it suffices to compute the intersection of $\hat{\cM}$ and $\cM_{\text{part}}$ in the claimed time bound.
In the following, we will set $\cM_1$ to be $\cM_{\text{part}}$ and $\cM_2$ to be $\hat{\cM}$ when talking about exchange graphs and other data structures.
Our main goal is to optimize the $O(n\sqrt{r})$ term to $O(r\sqrt{r})$, so it might be more intuitive to think of $r \ll n$.
We first show that for an $S \in \cI_\text{part} \cap \hat{\cI}$, the exchange graph $G(S)$ is quite unbalanced in the sense that most elements appear in the first distance layer.
In fact, the first distance layer of $G(S)$ contains all duplicates of elements $u$ in $U$ that do not appear in $S$.
This is by definition of $G(S)$ and the fact that $\cM_1 = \cM_{\text{part}}$ is the partition matroid.
In the following, when the context is clear, we let $d_t$ denote the $(s, t)$-distance of $G(S)$ and $L_1, \ldots, L_{d_t - 1}$ denote the distance layers.

\begin{fact}
  It holds that $L_1 = \{(u, i) \mid (u, i) \in \hat{U}\;\text{and}\;(u, j) \not\in S\;\text{for any $1 \leq j \leq k$}\}$.
  \label{fact:large-first-layer}
\end{fact}

Similarly, the odd layers of $G(S)$ (that corresponds to $\bar{S}$) are well-structured in the sense that they consist of elements whose one of the duplicates appears in $S$.
By definition of $G(S)$, we also know that elements in odd layers have only a single in-edge, which is from their corresponding duplicate in $S$.
These elements thus all have the same distance from $s$.

\begin{fact}
  It holds that $L_3 \cup L_5 \cup \cdots \cup L_{d_t - 1} = \{(u, i) \mid (u, i) \in \hat{U}\;\text{and}\;(u, j) \in S\;\text{for some}\;i \neq j\}$,
  and for each $(u, i) \in L_3 \cup \cdots \cup L_{d_t - 1}$, we have $d_{G(S)}(s, (u, i)) = d_{G(S)}(s, (u, j)) + 1$ where $(u, j) \in S$.
  \label{fact:structured-odd-layers}
\end{fact}

\paragraph{Union Exchange Graph.} Given the above facts, we introduce another notion of exchange graphs which is commonly used for matroid union (see, e.g., \cite{edmonds1968matroid,cunningham1986improved}).
For the given $k$ matroids $\cM_i = (U_i, \cI_i)$ and a subset $S \subseteq U$ that can be partitioned into $k$ independent sets $S_1, \ldots, S_k$ with $S_i \in \cI_i$, the \emph{union exchange graph} is a directed graph $H(S) = (U \cup \{s, t\}, E)$ with two distinguished vertices $s, t \not\in U$ and edge set $E = E_s \cup E_t \cup E_{\text{ex}}$, where

\begin{align*}
  E_s &= \{(s, u) \mid u \not\in S\}, \\
  E_t &= \{(u, t) \mid S_i + u \in \cI_i\;\text{for some $1 \leq i \leq k$}\},\;\text{and} \\
  E_{\text{ex}} &= \{(u, v) \mid S_i - v + u \in \cI_i\;\text{where $v \in S_i$}\}. \\
\end{align*}

We can see that the exchange graph $G(S)$ with respect to $S \in \cI_{\text{part}} \cap \hat{\cI}$ (as a subset of $\hat{U}$) and the union exchange graph $H(S)$ with respect to $S \subseteq U$ is essentially the same in the sense that $H(S)$ can be obtained from $G(S)$ by contracting all copies of the same element in the first layers and skipping all other odd layers.
In particular, for each $(u, i) \in S$, in $G(S)$, there might be a direct edge from $(u, i)$ to $(u, j)$ and an edge from $(u, j)$ to $(v, j)$, where $(v, j) \in S_j$ and $S_j - v + u \in \cI_j$.
Correspondingly, in $H(S)$, we skip the intermediate vertex $(u, j)$ and meld the above two edges as one direct edge from $u \in S_i$ to $v \in S_j$.
We also merge all edges from $s$ to some $(u, i)$ of the same $u$ in the first layer to a single edge from $s$ to $u$ (\cref{fact:large-first-layer}).
This simplification does not impact the distance layers of $H(S)$ since all such $(u, j)$ have the same distance from $s$ (\cref{fact:structured-odd-layers}).

From now on, for simplicity, our algorithms will run on the union exchange graphs $H(S)$, i.e., we will perform blocking-flow computation and augment $S$ along paths in $H(S)$.
On the other hand, to not repeat and specialize all the lemmas to the case of union exchange graphs, proofs and correctness will be argued implicitly in the perspective of the exchange graph $G(S)$ for matroid intersection.
For instance, for $P = (s, a_1, \ldots, a_{d_t - 1}, t)$ a shortest $(s, t)$-path in $H(S)$, ``augmenting $S$ along $P$'' means moving $a_i$ to the independent set that originally contains $a_{i + 1}$ for each $i \geq 1$, and thus effectively enlarge the size of $S$ by one via putting $a_1$ in it.\footnote{One can show that the matroid union $\cM$ is a matroid~\cite[Chapter 42]{schrijver2003}. As such, a basis can be obtained by trying to include each element into $S$. From the union exchange graph perspective, the independence test of $S + x$ corresponds to asking whether ``there is a path in $H(S)$ from $x$ to $t$''.}
One can verify that this is indeed what happens if we map $P$ back to a path $P^\prime$ in $G(S)$, and then perform the augmentation of $S$ (as a subset of $\hat{U}$) along $P^\prime$.

Our main idea to speed up the matroid union algorithm to $\tO_k(r\sqrt r)$ (instead of $\tO_k(n \sqrt{r})$) is to ``sparsify'' the first layer of $H(S)$ by only considering a subset of elements contained in some basis.  We formalize this in the following \cref{lemma:bfs-union,lemma:blocking-flow-union} together with \cref{alg:bfs-union,alg:blocking-flow-union}.

\begin{lemma}
  Given $S \in \cI_\text{part} \cap \hat{\cI}$ and $k$ bases $\{B_i\}_{i = 1}^{k}$ of $U_i \setminus S$, it takes $\tO_k(r)$ time to construct the distance layers $L_2, \ldots, L_{d_t - 1}$ of $H(S)$.
  \label{lemma:bfs-union}
\end{lemma}

Note that we know exactly what elements are in the first distance layer, so computing $L_2, \ldots, L_{d_t - 1}$ suffices.

\begin{algorithm}
  \SetEndCharOfAlgoLine{}
  \SetKwInput{KwData}{Input}
  \SetKwInput{KwResult}{Output}
  
  \caption{BFS in a union exchange graph}
  \label{alg:bfs-union}
  \KwData{$S \subseteq U$ which partitions into $S_1, \ldots, S_k$ of independent sets and $k$ bases $\{B_i\}_{i = 1}^{k}$ of $U_i \setminus S$}
  \KwResult{The $(s, u)$-distance $d(u)$ in $H(S)$ for each $u \in S \cup \{t\}$}

  $\mathsf{queue} \gets B_1 \cup \cdots \cup B_k$\;
  $d(u) \gets \infty$ for each $u \in S \cup \{t\}$, and $d(u) \gets 1$ for each $u \in B_1 \cup \cdots \cup B_k$\;
  $\cT_i \gets \textsc{Initialize}(\cM_i, S_i, S_i)$ (\cref{thm:bst} with $\beta = 1$)\;
  \While{$\mathsf{queue} \neq \emptyset$} {
    $u \gets \mathsf{queue}.\textsc{Pop}()$\;
    \For{$i \in \{1, 2, \ldots, k\}$ where $u \in U_i$ and $u \not\in S_i$} {
      \While{$v := \cT_i.\textsc{Find}(u) \neq \bot$} {
        $d(v) \gets d(u) + 1$ and $\mathsf{queue}.\textsc{Push}(v)$\;
        $\cT_i.\textsc{Delete}(v)$\;
      }
      \lIf{$S_i + u \in \cI$ and $d(t) = \infty$} {
        $d(t) \gets d(u) + 1$
      }
    }
  }
  \textbf{return} $d(u)$ for each $u \in S \cup \{t\}$\;
\end{algorithm}

\begin{proof}
  The algorithm is presented as \cref{alg:bfs-union}, and it is essentially a breadth-first-search (BFS) starting from $B_1 \cup \cdots \cup B_k$ instead of $s$.
  Out-edges in $H(S)$ are explored via $k$ binary search trees $\cT_1, \cT_2, \ldots, \cT_k$ of \cref{thm:bst}, one for each matroid $\cM_i$ and independent set $S_i$.
  Let's analyze the running time first.
  Building $\cT_i$ takes a total of $\tO(|S|) = \tO_k(r)$ time.
  Exploring the graph takes $\tO(|S \cup B_1 \cup \cdots \cup B_k| \cdot k) = \tO_k(r)$ time in total since each element in $S$ is found at most once by $\cT_i.\textsc{Find}(\cdot)$ because $S_i$'s are disjoint, and we also spend $O(k)$ time for each element in $S \cup B_1 \cup \cdots \cup B_k$ iterating over $\cT_i$.
  
  It remains to show that starting from $B_1 \cup \cdots \cup B_k$ instead of $U \setminus S$ does not affect the correctness of the BFS.
  For this, it suffices to show that we successfully compute $d(u)$ for all $u \in S$ with distance $2$ from $s$.
  By definition, $u \in S_i$ is of distance $2$ from $s$ if and only if there exists an $x \in U_i \setminus S$ such that $S_i - u + x \in \cI_i$.
  This is equivalent to $\rank_i(S_i - u + (U_i \setminus S)) > \rank_i(S_i)$ by \cref{obs:exchange}.
  But then by \cref{lemma:basis-rank}, we have $\rank_i(S_i - u + (U_i \setminus S)) = \rank_i(S_i - u + B_i)$, and so such an $x$ exists in $B_i$ as well.
  This concludes the proof of \cref{lemma:bfs-union}.
\end{proof}

\begin{lemma}
  Given an $S \in \cI_{\text{part}} \cap \hat{\cI}$ with $d_{H(S)}(s, t) = d_t$ together with data structures $\cD_i$ of \cref{thm:decremental-basis} that maintains a basis of $U_i \setminus S$ for each $1 \leq i \leq k$, it takes $\tO_k(r + \frac{r\sqrt{r}}{d_t} + (|S^\prime| - |S|) \cdot d_t\sqrt{r})$ time to obtain an $S^\prime \in \cI_{\text{part}} \cap \hat{\cI}$ with $d_{H}(S^\prime)(s, t) > d_t$, with an additional guarantee that $\cD_i$ now maintains a basis of $U_i \setminus S^\prime$ for each $1 \leq i \leq k$.
  \label{lemma:blocking-flow-union}
\end{lemma}

\begin{algorithm}
  \SetEndCharOfAlgoLine{}
  \SetKwInput{KwData}{Input}
  \SetKwInput{KwResult}{Output}
  \SetKwInput{KwGuarantee}{Guarantee}
  
  \caption{Blocking flow in a union exchange graph}
  \label{alg:blocking-flow-union}
  \KwData{$S \subseteq U$ which partitions into $S_1, \ldots, S_k$ of independent sets and a dynamic-basis data structure $\cD_i$ of $U_i \setminus S$ for each $1 \leq i \leq k$}
  \KwResult{$S^\prime \in \cI_{\text{part}} \cap \cI_k$ with $d_{H(S^\prime)}(s, t) > d_{H(S)}(s, t)$}
  \KwGuarantee{$\cD_i$ maintains a basis of $U_i \setminus S^\prime$ at the end of the algorithm for each $1 \leq i \leq k$}
  
  Build the distance layers $L_2, \ldots, L_{d_t - 1}$ of $H(S)$ with \cref{lemma:bfs-union}\;
  $L_0 \gets \{s\}$ and $L_{d_t} \gets \{t\}$\;
  $B_i \gets$ the basis maintained by $\cD_i$ and $L_1 \gets B_1 \cup \cdots \cup B_k$\;
  $A_\ell \gets L_\ell$ for each $0 \leq \ell \leq d_t$\;
  $\cT_{\ell}^{(i)} \gets \textsc{Initialize}(\cM_i, S_i, Q_{S_i}, A_{\ell} \cap S_i)$ for each $2 \leq \ell < d_t$ and $1 \leq i \leq k$ (\cref{thm:bst} with $\beta = \sqrt{r} / d_t$)\;
  $D_{\ell} \gets \emptyset$ for each $1 \leq \ell < d_t$\;
  $\ell \gets 0$ and $a_0 \gets s$\;
  \While{$\ell \geq 0$} {
    \If{$\ell < d_t$} {
      \lIf{$A_{\ell} = \emptyset$} {
        \textbf{break}
      }
      \lIf{$\ell = 0$} {
        Find an $a_{\ell + 1} := \cT_{\ell + 1}^{(i)}.\textsc{Find}(a_\ell) \neq \bot$ for some $1 \leq i \leq k$
      }
      \lElse {
        $a_{\ell + 1} \gets$ an arbitrary element in $A_1$
      }
      \If{such an $a_{\ell + 1}$ does not exist} {
        \lIf{$\ell \geq 2$} {
          $\cT_{\ell}^{(j)}.\textsc{Delete}(a_{\ell})$ where $a_{\ell} \in S_j$
        }
        $A_\ell \gets A_\ell - a_\ell$ and $\ell \gets \ell - 1$\;
      }
      \Else {
        $\ell \gets \ell + 1$
      }
    }
    \Else {
    \tcp{Found augmenting path $a_1, a_2, \ldots a_\ell$}
      $D_1 \gets D_1 + a_1$ and $A_1 \gets A_1 - a_1$\;
      \For{$i \in \{1, 2, \ldots, k\}$ where $a_1 \in U_i$} {
        $B_i \gets B_i - a_1$\;
        \If{$\cD_i.\textsc{Delete}(a_1)$ returns a replacement $x$} {\label{line:delete}
          $B_i \gets B_i + x$ and $A_1 \gets A_1 \cup \{x\}$\;
        }
      }
      \For{$i \in \{2, \ldots, d_t - 1\}$} {
          $D_{i} \gets D_{i} + a_{i}$ and $A_{i} \gets A_{i} - a_{i}$\;
          $\cT_{i}^{(j)}.\textsc{Delete}(a_i)$ and $\cT_{i}^{(j)}.\textsc{Update}(\{a_{i - 1}, a_{i}\})$ where $a_i \in S_j$\;
      }
      Augment $S$ along $P = (s, a_1, \ldots, a_{d_t - 1}, t)$\;
      $\ell \gets 0$\;
    }
  }
  \textbf{return} $S$\;
\end{algorithm}

\begin{proof}[Proof of \Cref{lemma:blocking-flow-union}]
  Our blocking-flow algorithm for matroid union is presented as \cref{alg:blocking-flow-union}.
  As it is equivalent to \cref{alg:blocking-flow} running on $G(S)$ except that the first layer $L_1 := B_1 \cup \cdots \cup B_k$ is now only a subset (which is updated after each augmentation) of $\ground \setminus S$, we skip most parts of the proof and focus on discussing this difference.
  That is, we need to show that if $A_1$ becomes empty, then there is no augmenting path of length $d_t$ in $H(S^\prime)$ anymore.
  Given how $A_1$ and $B_i$'s are maintained and \cref{lemma:key} (note that the set $X$ maintained in $\cT_{\ell}^{(i)}$ is always $A_\ell \cap S_i$ with respect to the current $S_i$ and thus it lets us explore out-edges to $A_{\ell} \cap S_i$ satisfying \cref{eq:condition}), $A_1$ is always the subset of $B_1 \cup \cdots \cup B_k$ consisting of elements that still potentially admits augmenting path of length $d_t$ in $H(S^\prime)$ through them.
  That means if $A_1 = \emptyset$, then there is no augmenting set of length $d_t$ in $G(S^\prime)$, that starts from some $b \in B_1 \cup \cdots \cup B_k$.
  This would imply that there is no such path even if we start from $x \in (\ground_i \setminus S) \setminus D_1$ as $B_i$ is a basis of it: if $S_i + x - y \in \cI$ for some $x \in (\ground_i \setminus S) \setminus D_1$ and $y \in S_i$, then there is a $b \in B$ with $S + b - y \in \cI$, and thus a path starting from $x$ can be converted into a path starting from $b$.
  On the other hand, all elements in $D_1$ are not on a such path by \cref{lemma:monotone} either.
  This shows that indeed $d_{H(S^\prime)}(s, t) > d_{H(S)}(s, t)$.
  
  The guarantee that $\cD_i$ now operates on $U_i \setminus S^\prime$ is clear: Augmenting along $P = (s, a_1, \ldots, a_{d_t - 1}, t)$ corresponds to adding $a_1$ into $S$, and since we call $\cD_i.\textsc{Delete}(a_1)$ in Line~\ref{line:delete} after each such augmentation, $\cD_i$ indeed stays up-to-date. 
  
  It remains to analyze the running time of \cref{alg:blocking-flow-union}.
  Computing distance layers with \cref{lemma:bfs-union} takes $\tO_k(r)$ time.
  The number of elements that have ever been in some $A_i$ is $O_k(r + |S^\prime| - |S|)$ since
  (i) $L_2 \cup \cdots \cup L_{d_t - 1}$ has size $O_k(r)$,
  (ii) the initial basis $B_i$ of $\ground_i \setminus S$ for each $1 \leq i \leq k$ has total size $O_k(r)$, and
  (iii) each of the $|S^\prime| - |S|$ augmentations adds at most $O_k(1)$ elements to $A_1$.
  Similar to \cref{lemma:blocking-flow}, this means that there are at most $O_k(r)$ iterations, each taking $O_k(\frac{\sqrt{r}}{d_t})$ time in $\cT_{\ell + 1}^{(i)}.\textsc{Find}(\cdot)$ with our choice of $\beta$.
  The algorithm found $|S^\prime| - |S|$ augmenting paths, taking $\tO_k(d_t\sqrt{r} \cdot (|S^\prime| - |S|))$ time in total to update the binary search trees.
  Also, for each such augmentation, we need $\tO_k(\sqrt{r})$ time to update the basis $B_i$ for all $1 \leq i \leq k$, which is subsumed by the cost of updating $\cT_{i}^{(j)}$.
  These components sum up the total running time of
  \[ \tO_k\left(r + \frac{r\sqrt{r}}{d_t} + \left(|S^\prime| - |S|\right) \cdot d_t\sqrt{r}\right). \]
\end{proof}

\cref{thm:dynamic-matroid-union} now follows easily.

\begin{proof}[Proof of \cref{thm:dynamic-matroid-union}]
  We initialize the dynamic-basis data structure $\cD_i$ of \cref{thm:decremental-basis} on $\ground_i$ for each of the matroid $\cM_i$.
  We then run \cref{lemma:blocking-flow-union} for at most $\sqrt{r}$ iterations with $\{\cD_i\}_{i = 1}^{k}$ until $d_{H(S)}(s, t) \geq \sqrt{r}$ and get an $S \in \cI_{\text{part}} \cap \hat{\cI}$ with $\cD_i$ now operating on $\ground_i \setminus S$ for each $1 \leq i \leq k$.
  This takes
  \[ \tO_k\left(r\sqrt{r} + \sum_{d = 1}^{\sqrt{r}}\frac{r\sqrt{r}}{d} + \sum_{d = 1}^{\sqrt{r}}d \cdot \left(|S_d| - |S_{d - 1}|\right)\right) = \tO_k(r\sqrt{r})\]
  time.
  By \cref{lemma:approx}, $S$ is $O_k(\sqrt{r})$ steps away from being optimal, and thus we find the remaining augmenting paths one at a time using \cref{lemma:bfs-union} in $\tO_k(r\sqrt{r})$ time in total.
  Note that since a single augmentation corresponds to adding an element to $S$ (hence removing it from $\ground \setminus S$), we can maintain the basis of $\ground_i \setminus S$ that \cref{lemma:bfs-union} needs in $\tO_k(\sqrt{r} \cdot \sqrt{r})$ total update time, which is subsumed by other parts of the algorithm.
\end{proof}

\subsection{Matroid Packing and Covering}
\label{sec:packing}

A direct consequence of our matroid union algorithm (\cref{thm:dynamic-matroid-union-fold} in particular) is that we can solve the following packing and covering problem efficiently.
As a reminder, the exact dependence on $k$ of our algorithm is $\tO(n + kr\sqrt{\min(n, kr)} + k \min(n, kr))$ by \cref{thm:dynamic-matroid-union-fold}.

\begin{restatable}[Packing]{corollary}{packingalgo}
    For a matroid $\cM = (\ground, \cI)$, it takes $\tO(n\sqrt{n} + \frac{n^2}{r})$ time to find the largest integer $k$ and a collection of disjoint subsets $\cS = \{S_1, S_2, \ldots, S_k\}$ of $\ground$ such that $S_i$ is a basis for each $1 \leq i \leq k$ under the dynamic-rank-query model.
    \label{cor:packing}
\end{restatable}

\begin{proof}
  It's obvious that $k \leq \frac{n}{r}$ holds.
  We do a binary search of $k$ in the range $[0, \frac{n}{r}]$, and for each $k$, we can determine the largest subset $S$ of $\ground$ which can be partitioned into $k$ disjoint independent sets by \cref{thm:dynamic-matroid-union-fold}.
  If $|S| = kr$, then it means that there are at least $k$ disjoint bases.
  Otherwise, there are less than $k$ disjoint bases.
  The running time is $\tO(n\sqrt{n} + \frac{n^2}{r})$.
\end{proof}

\begin{restatable}[Covering]{corollary}{coveringalgo}
    For a matroid $\cM = (\ground, \cI)$, it takes $\tO(\alpha r\sqrt{n} + \alpha n)$ time to find the smallest integer $\alpha$ and a partition $\cS = \{S_1, S_2, \ldots, S_\alpha\}$ of $\ground$ such that $S_i \in \cI$ holds for each $1 \leq i \leq \alpha$ under the dynamic-rank-query model.
    \label{cor:covering}
\end{restatable}

\begin{proof}
  We first obtain a $2$-approximation $\alpha^{\prime}$ of $\alpha$ (i.e., $\alpha \leq \alpha^\prime \leq 2\alpha$) by enumerating powers of $2$, running \cref{thm:dynamic-matroid-union-fold} with $k = 2^i$, and checking if the returned $S$ has size $n$: If $|S| = n$, then we know $2^i$ independent sets suffice to cover $\ground$.
  Note that the enumeration stops whenever we found a suitable value of $\alpha^\prime$.
  The exact value of $\alpha$ can then be found by a binary search in $[\frac{\alpha^\prime}{2}, \alpha^\prime]$.
  This takes $\tO(\alpha r\sqrt{n} + \alpha n)$ (note that $\alpha r \geq n$ must hold).
\end{proof}

\subsection{Application: Spanning Tree Packing}
\label{sec:spanning}

We demonstrate the applicability of our techniques by deriving an $\tO(|E| + (k|V|)^{3/2})$ algorithm for the $k$ disjoint spanning tree problem in a black-box manner.
This improves Gabow's specialized $\tO(k^{3/2}|V|\sqrt{|E|})$ algorithm~\cite{gabow1988forests}.
Since all applications of our algorithms follow the same reduction, we only go through it once here.
Refer to \cref{appendix:applications} for other applications of both our matroid union and matroid intersection algorithms.

\begin{theorem}
  Given an undirected graph $G = (V, E)$, it takes $\tO(|E| + (k|V|)^{3/2})$ time to find $k$ edge-disjoint spanning trees in $G$ or determine that such spanning trees do not exist with high probability\footnote{We use \emph{with high probability} to denote with probability at least $1 - |V|^{-c}$ for an arbitrarily large constant $c$.}.
  \label{thm:k-dst}
\end{theorem}

\begin{proof}
By \cref{thm:dynamic-matroid-union-fold}, it suffices to provide a data structure that supports the three dynamic-oracle operations (\cref{def:dyn-oracle}) in $\polylog(|V|)$ time.
Our black-box reduction makes use of the worst-case connectivity data structure of \cite{KapronKM13,GibbKKT15}, which can be adapted to in $O(\polylog(|V|))$ update time maintain the rank of a set of edges (see \cref{appendix:applications-graphic} for a discussion on how this can be done).

  Let $\cM_G$ be the graphic matroid with respect to $G = (V, E)$.
  $G$ admits $k$ edge-disjoint spanning trees if and only if $\cM_G$ admits $k$ disjoint bases.
  The theorem now follows from \cref{thm:dynamic-matroid-union-fold} with $n = |E|$ and $r = |V| - 1$ since \cref{thm:dynamic-matroid-union-fold} returns a union of $k$ disjoint bases if they exist (we note that $k \le |E|/(|V|-1) \le O(|V|)$, and hence the $O(k^2 r)$ term is dominated by the $O((kr)^{3/2})$ term).
\end{proof}

\section{Super-Linear Query Lower Bounds} \label{sec:lowerbound}

Lower bounds for matroid intersection have been notoriously difficult to prove. The current highest lower bound is due to Harvey \cite{harvey2008matroid} which says that $(\log_2 3) n - o(n)$ queries are necessary for any deterministic independence-query algorithm solving matroid intersection.
Obtaining an $\omega(n)$ lower bound has been called a challenging open question \cite{chakrabarty2019faster}.

In this section, we show the first super-linear query lower bound for matroid intersection, both in our \emph{new dynamic-rank-oracle} model (\cref{def:dyn-oracle}), and also for the \emph{traditional independence-oracle} model, thus answering the above-mentioned open question and improving on the bounds of \cite{harvey2008matroid}.
We obtain our lower bounds by studying the communication complexity for matroid intersection.

\begin{theorem}
If Alice is given a matroid $\cM_1 = (\ground,\cI_1)$
and Bob a matroid $\cM_2 = (\ground,\cI_2)$, any deterministic communication protocol needs $\Omega(n \log n)$ bits of communication to solve the matroid intersection problem.
\label{thm:comm-lb}
\end{theorem}

The communication lower bound of \cref{thm:comm-lb} implies a similar lower bound for the number of independence queries needed. 
We argue that any independence-query algorithm can be simulated by Alice and Bob in the communication setting by exchanging a single bit per query asked. Whenever they want to ask an independence query ``Is $S \in \cI_{i}$?'', Alice or Bob will check this locally and share the answer with the other party by sending one bit of communication.

Unfortunately, this argument does not extend to the traditional rank-oracle model (since each rank query can in fact reveal $\Theta(\log n)$ bits of information, which need to be sent to the other party). However, for the new \emph{dynamic}-rank-oracle model, the $\Omega(n\log n)$ lower bound holds as now each new query only reveals constant bits of information: either the rank remains the same, increases by one, or decreases by one (and Alice or Bob can send which is the case to the other party with a constant number of bits). Our discussion proves the following corollaries, given \cref{thm:comm-lb}.

\begin{corollary}
Any deterministic (traditional) independence-query algorithm solving matroid intersection requires $\Omega(n \log n)$ queries.
\end{corollary}

\begin{corollary}
Any deterministic dynamic-rank-query algorithm solving matroid intersection requires $\Omega(n \log n)$ queries.
\end{corollary}

\begin{remark}
We note that our lower bounds are also valid for the \emph{matroid union} problem, due to the standard reductions\footnote{See \cref{sec:reduction} for a reduction from matroid union to matroid intersection. To reduce from matroid intersection to matroid union, consider $\cM = \cM_1 \vee \cM_2^{*}$, where $\cM_2^{*}$ is the \emph{dual matroid} of $\cM_2$ ($S \subseteq U$ is independent in $\cM_2^{*}$ if and only if $U - S$ contains a basis). It's easy to show that the basis $B$ of $\cM$ in $\cM_1$ will be of the form $B = S \cup (U \setminus R)$, where $S$ is the solution to the intersection between $\cM_1$ and $\cM_2$ and $R$ is an arbitrary basis of $\cM_2$ that contains $S$.} between matroid intersection and union.
\end{remark}

\subsection{Communication Setting}

We study the following communication game
which we call \textsf{Matroid-Intersection-with-Candidate}.
Alice and Bob are given matroids $\cM_1 = (\ground,\cI_1)$ respectively $\cM_2 = (\ground,\cI_2)$.
Suppose they are also both given a common independent set $S\in \cI_1\cap \cI_2$, and they wish to determine
whether $S$ is a maximum-cardinality independent
set.
Clearly 
\textsf{Matroid-Intersection-with-Candidate} is an easier version of
the matroid intersection problem, as Alice and Bob can just ignore the candidate $S$.

Our idea is that in order to solve
\textsf{Matroid-Intersection-with-Candidate}, Alice and Bob need to determine if there exists an \emph{augmenting path}---that is an $(s,t)$-path---in the \emph{exchange graph} $G(S)$ (see \cref{def:exchange-graph} and \cref{lemma:augmenting-path}). It is known that \textsf{$(s,t)$-connectivity} in a
graph requires $\Omega(n\log n)$ bits of communication (\cref{lem:lb-conn}, \cite{HajnalMT88}). Using \emph{strict gammoids} as our matroids, we argue that we can choose exactly how the underlying exchange graph looks like, and hence that matroid intersection admits the same lower bound.

\begin{definition}[Strict Gammoid, see~\cite{perfect1968applications,mason1972class}]
Let $H = (V,E)$ be a directed graph and $X\subseteq V$ a subset of vertices.
Then $(H,X)$ defines a matroid $\cM = (V,\cI)$ called a \emph{strict gammoid}, where a set of vertices $Y\subseteq V$ is independent if and only if
 there exists a set of vertex-disjoint directed paths (some of which might just consist of single vertices) in $H$ whose starting points all belong to $X$ and whose ending points are exactly $Y$.
\end{definition}

\begin{claim}
\label{clm:lb-matroids}
Suppose $G=(L,R,E)$ is a directed bipartite graph and $a, b \in R$ are two unique vertices such that $a$ has zero in-degree and $b$ has zero out-degree.
Then there exist two matroids $\cM_1, \cM_2$ over the ground set $L\cup R$ such that
$L$ is independent in both matroids and
the \emph{exchange graph} $G(L)$ is exactly $G$ plus two extra vertices ($s$ and $t$) and two extra edges ($s\to a$ and $b \to t$).
\end{claim}

\begin{proof}
Let $F_{1} = \{(u,v) | (u,v)\in E, u\in L, v\in R\}$ be the directed edges from $L$ to $R$ in $G$,
and
$F_{2} = \{(u,v) | (v,u)\in E, u\in L, v\in R\}$ be the (reversed) directed edges from $R$ to $L$ in $G$.
Also let $H_1 = (L\cup R,F_1)$ and $H_2 = (L\cup R,F_2)$ be the directed graphs with these edges respectively.

We let $\cM_1$, respectively $\cM_2$, be the strict gammoids defined by $(H_1, L+a)$ respectively $(H_2, L+b)$. Now $L$ is independent in both matroids. It is straightforward to verify that the exchange graph $G(L)$ is exactly as described in the claim. We certify that this is the case for the edges defined by $\cM_1$ ($\cM_2$ is similar):

\begin{enumerate}
    \item  $G(L)$ will have an edge from $s$ to $a$, since
    $L+a$ is independent in $\cM_1$. 
    Additionally note that $a$ has in-degree zero in $G$ (and hence is an isolated vertex in $H_1$).
    
    \item
    For any $x\in L, y\in R$, the edge $(x,y)$ exists in $G(L)$
    if and only if $L - x + y$ is independent in $\cM_1$. By definition this is if and only if there exists a vertex-disjoint path starting from $L$ and ending to $L-x+y$ in $H_1$, or equivalently if the edge $(x,y)$ exists in $H_1$ (indeed, all vertices in $L-x$ must be both starts and ends of paths, so the path to $y$ must have started in $x$).
    \qedhere{}
    
\end{enumerate}

\end{proof}

We now proceed to reduce an instance of \textsf{$(s,t)$-connectivity} to that of \textsf{Matroid-Intersection-with-Candidate}, which concludes the proof of \cref{thm:comm-lb}.

\begin{definition}[\textsf{$(s,t)$-connectivity}]
Suppose $G = (V, E_A \cup E_B)$ is an undirected graphs on $n=|V|$ vertices,
where Alice knows edges $E_A$ and Bob knows edges $E_B$.
They are also both given vertices $s$ and $t$, and want to determine
if $s$ and $t$ are connected in $G$.
\label{def:st-conn}
\end{definition}

\begin{lemma}[\cite{HajnalMT88}]
\label{lem:lb-conn}
The deterministic communication complexity of
\textsf{$(s,t)$-connectivity} is $\Omega(n\log n)$.
\end{lemma}

\begin{proof}[Proof of \cref{thm:comm-lb}.]
We show that an instance of \textsf{$(s,t)$-connectivity} can be converted to an instance
of \textsf{Matroid-Intersection-with-Candidate} of roughly the same size.
Suppose the symbols are defined as in  Definition \ref{def:st-conn}.
Let $\bar{V} = \{\bar{v} : v\in V\}$ be a copy of $V$.
We construct a directed bipartite graph $G' = (V, \bar{V}, E'_A \cup E'_B)$
as follows:

\begin{itemize}
    \item $(v,\bar{v}) \in E'_A$ for all $v\in V$.
    \item $(\bar{v},v) \in E'_B$ for all $v\in V$.
    \item $(v,\bar{u}), (u,\bar{v}) \in E'_A$ for all $\{u,v\}\in E_A$.
    \item $(\bar{v},u), (\bar{u},v) \in E'_A$ for all $\{u,v\}\in E_B$.
    \item No other edges exist.
\end{itemize}

Alice knows $E'_A$, and Bob knows $E'_B$. $G'$ has $2n$ vertices and $2n+2m$ edges.

Now let $G''$ be $G'$ but removing all incoming edges from $\bar{s}$ and
all outgoing edges from $\bar{t}$, in order to apply \cref{clm:lb-matroids} on $G''$ with $a = \bar{s}$ and $b = \bar{t}$.
Say we get matroids $\cM_1$ and $\cM_2$. Note that 
Alice knows $\cM_1$ and Bob knows $\cM_2$ by construction.

Now $s$ and $t$ are connected in $G$ if and only if there is a directed $(\bar{s},\bar{t})$-path in $G''$. This happens if and only if $V$ is not a maximum-cardinality common independent set of $\cM_1$ and $\cM_2$ (i.e.\ in the case we found an augmenting path for $V$).

Hence if there is a (deterministic) communication protocol for matroid intersection using $c$ bits of communication, there is also one for \textsf{$(s,t)$-connectivity} using $O(c)$ bits of communication. \cref{lem:lb-conn} then implies the
$\Omega(n\log n)$ communication lower bound for matroid intersection.
\end{proof}

\section{Open Problems} 
\label{sec:openproblems}

Our dynamic-oracle model opens up a new path to achieve fast algorithms for many problems at once, where the ultimate goal is to achieve near-linear time and dynamic-rank-query complexities. This would imply near-linear time algorithms for many fundamental problems. We envision reaching this goal via a research program where the studies of algorithms and lower bounds in our and the traditional models 
complement each other. 
In particular, a major open problem is to improve our algorithms further, which would imply improved algorithms for many problems simultaneously.
A major step towards this goal is improved algorithms in the traditional model, which would already be a breakthrough. Moreover, failed lower bound attempts might lead to new algorithmic insights and vice versa, and we leave improving our lower bounds as another major open problem.
We believe that the communication complexity of graph and matroid problems is an important component in this study since it plays a main role in our lower bound argument. Recently the communication and some query complexities of bipartite matching and related problems were resolved in \cite{blikstadBEMN22}. How about the communication and query complexities of dynamic-oracle matroid problems and their special cases such as colorful spanning trees?
It is also fruitful to resolve some special cases as the solutions may shed more light on how to solve matroid problems in our model. Below are some examples.

\begin{itemize}
  \item \textbf{Disjoint Spanning Trees.} Can we find $k$ edge-disjoint spanning trees in an undirected graph in near-linear time for constant $k$, or even do so for the case of $k = 2$ (which already has application in the Shannon Switching Game)? Our new $\tO(|E|+|V|\sqrt{|V|})$-time algorithm shows that it is possible for sufficiently dense graphs. For the closely related problem of finding $k$ edge-disjoint arborescences (rooted \emph{directed} spanning trees) in a directed graph, the case of $k = 2$ has long been settled by Tarjan's linear time algorithm~\cite{Tarjan76}, and the case of constant $k$ has also been resolved by~\cite{BhalgatHKP08}. It is a very interesting question whether the directed case is actually computationally easier than the undirected case or not.

  \item \textbf{Colorful Spanning Tree.} This problem generalizes the maximum bipartite matching problem, among others. Given the recent advances in max-flow algorithms which are heavily based on continuous optimization techniques, bipartite matching can now be solved in almost-linear time~\cite{ChenKLPGS22} in general and nearly linear time for dense input graphs \cite{BrandLNPSSSW20}. 
  It is very unclear if continuous optimization can be used for colorful spanning tree since its linear program has exponentially many constraints. 
  This reflects the general challenge of using continuous optimization to solve matroid problems and many of their special cases. 
  Thus, improving  Hopcroft-Karp's $O(|E|\sqrt{|V|})$ runtime \cite{HopcroftK73} (which is matched by our dynamic-oracle matroid algorithm) may shed some light on either how to use continuous optimization for these problems or how combinatorial algorithms can break this runtime barrier for colorful spanning tree, bipartite matching, and matroid problems. 

\end{itemize}

\paragraph{Other Problems with Dynamic Oracles.} It also makes sense to define dynamic oracles for problems like submodular function minimization (SFM), which asks to find the minimizer of a submodular function given an evaluation oracle. In this regime, similar to matroid intersection, we want to limit the symmetric difference from the current evaluation query to the previous ones. 
We believe that the recent algorithms for submodular function minimization based on convex optimization and cutting-plane methods, particularly the work of \cite{LeeSW15,JiangLSW20}, can be adapted to the dynamic-oracle setting.  
However, 
we are not aware of any applications of these dynamic-oracle algorithms. 
The first step is thus improving the best bounds in the traditional oracle model. 
The special case of the cut-query setting \cite{RubinsteinSW18,MukhopadhyayN20,LeeSZ21,LeeLSZ21,ApersEGLMN22} is also very interesting; we leave getting algorithms for min $(s,t)$-cut \cite{ChenKLPGS22} and directed global mincut \cite{Cen0NPSQ21} with near-linear time and dynamic-query complexity as major open problems.\footnote{Adapting the cut-query algorithm of \cite{MukhopadhyayN20} to work with dynamic cut oracles and, even better, with a parallel algorithm \cite{AndersonB21,Lopez-MartinezM21}, is also open; though, we suspect that these are not hard.}
Another interesting direction is the quantum setting. For example, can one define the notion of dynamic quantum cut query so that the quantum cut-query algorithm of \cite{ApersEGLMN22} can imply a non-trivial quantum global mincut algorithm?

\paragraph{Improved Lower Bounds.} 
Obtaining improved lower bounds for matroid intersection is also an important open problem. Getting $\Omega(n\log n)$ lower bound for traditional rank-query matroid intersection algorithms is particularly interesting since it would subsume our $\Omega(n\log n)$ lower bounds (traditional rank-query lower bound implies independence-query and dynamic-rank-query lower bounds) and the $\Omega(n \log n)$ SFM lower bound of \cite{ChakrabartyGJS22}. 
For the latter, \cite{ChakrabartyGJS22} showed an $\Omega(n\log{n})$ lower bound for SFM against \emph{strongly}-polynomial time algorithms. Since SFM generalizes matroid intersection in the traditional rank-oracle model (i.e., a rank query of a matroid corresponds to an evaluation of the submodular function), getting the same lower bound for traditional rank-query matroid intersection algorithms would further strengthen the result of \cite{ChakrabartyGJS22} to hold against \emph{weakly}-polynomial time algorithms. 

Additionally, achieving a {\em truly} super-linear lower bound (i.e. an $n^{1+\Omega(1)}$ bound) for any of the above problems is extremely interesting.

\bibliography{Bibliography}

\appendix

\section{\texorpdfstring{$k$}{k}-Fold Matroid Union} \label{appendix:matroid-union-fold}

In this section, we study the special case of matroid union where we take the $k$-fold union of the same matroid.
That is, a basis of the $k$-fold union of $\cM = (U, \cI)$ is the largest subset $S \subseteq U$ which can be partitioned into $k$ disjoint independent sets $S_1, \ldots, S_k$ of $\cI$.
Many of the prominent applications of matroid union fall into this special case, particularly the $k$-disjoint spanning tree problem.
As a result, here we show an optimized version of the algorithm presented in \cref{sec:matroid-union} with better and explicit dependence on $k$ that works in this regime. 

\matroidunionfold*

Note that in the breadth-first-search and blocking-flow algorithms in \cref{sec:matroid-union}, there is an $O(k)$ overhead where we have to spend $O(k)$ time iterating through the $k$ binary search trees in order to explore the out-neighbors of the $O(kr)$ elements.
Our goal in this section is thus to show that it is possible to further ``sparsify'' the exchange graphs to contain essentially only a basis, hence reducing its size from $\Theta(kr)$ to $O(r)$.
We start with a slight modification to the BFS \cref{alg:bfs-union} which reduces the running time by a factor of $O(k)$. The idea is that if we visit an element $u$ in the BFS which does not increase the rank of all visited elements so far, we can skip searching out-edges from $u$. Indeed, if $(u,v)$ is an edge of the exchange graph, then there must have been some element $u'$ visited earlier in the BFS which also has the edge $(u',v)$.

\begin{algorithm}
  \SetEndCharOfAlgoLine{}
  \SetKwInput{KwData}{Input}
  \SetKwInput{KwResult}{Output}
  
  \caption{BFS in a $k$-fold union exchange graph}
  \label{alg:bfs-union-fold}
  \KwData{$S \subseteq U$ with partition $S_1, \ldots, S_k$ of independent sets and a basis $B$ of $U \setminus S$}
  \KwResult{The $(s, v)$-distance $d(v)$ in $H(S)$ for each $v \in S \cup \{t\}$}

  $\mathsf{queue} \gets B$ and $R \gets \emptyset$\;
  $d(v) \gets \infty$ for each $v \in S \cup \{t\}$, and $d(v) \gets 1$ for each $v \in B$\;
  $\cT_i \gets \textsc{Initialize}(\cM, S_i, S_i)$ (\cref{thm:bst} with $\beta = 1$)\;
  \While{$\mathsf{queue} \neq \emptyset$} {
    $u \gets \mathsf{queue}.\textsc{Pop}()$\;
    \If{$R + u \in \cI$} {
      \For{$i \in \{1, 2, \ldots, k\}$} {
        \While{$v := \cT_i.\textsc{Find}(u) \neq \bot$} {
          $d(v) \gets d(u) + 1$ and $\mathsf{queue}.\textsc{Push}(v)$\;
          $\cT_i.\textsc{Delete}(v)$\;
        }
        \lIf{$S_i + u \in \cI$ and $d(t) = \infty$} {
          $d(t) \gets d(u) + 1$
        }
      }
      $R \gets R + u$\;
    }
  }
  \textbf{return} $d(v)$ for each $v \in S \cup \{t\}$.
\end{algorithm}

\begin{lemma}
  Given $S \in \cI_\text{part} \cap \hat{\cI}$ and a basis $B$ of $U \setminus S$, it takes $\tO(kr)$ time to construct the distance layers $L_2, \ldots, L_{d_t - 1}$ of $H(S)$.
  \label{lemma:bfs-union-fold}
\end{lemma}

\begin{proof}
  The algorithm is presented as \cref{alg:bfs-union-fold}.
  It performs a BFS from a basis $B$ of the first layer and only explores out-edges from the first basis $R$ it found.
  It takes
  (i) $\tO(|S|)$ time to construct the $\cT_i$'s,
  (ii) $\tO(1)$ time to discover each of the $O(kr)$ element, and
  (iii) an additional $\tO(k \cdot |R|)$ time to iterate through all $k$ binary search trees $\cT_i$'s for each $u \in R$.
  The total running time is thus bounded by $\tO(kr)$.
  
  We have shown in \cref{lemma:bfs-union} that it is feasible to replace $U \setminus S$ with simply $B$.
  It remains to show that exploring only the out-neighbors of $u \in R$ does not affect the correctness.
  Consider a $v \in S \setminus R$ (we know that $B \subseteq R$ so it suffices to consider elements in $S$) with an out-neighbor $x \in S_i$, i.e., $S_i - x + v \in \cI$.
  It then follows that $\rank((S_i - x) + R_v) = \rank((S_i - x) + (R_v + v)) \geq \rank(S_i)$ by \cref{obs:exchange} and \cref{lemma:basis-rank}, where $R_v$ is the set $R$ when $v$ is popped out of the queue (in other words, $R_v + v \not\in \cI$).
  This implies that there is a $u \in R_v$ which is visited before $v$ that also has out-neighbor $x$.
  The modification is therefore correct.
\end{proof}

Our blocking-flow algorithm for $k$-fold matroid union is presented as \cref{alg:blocking-flow-union-fold}.
It's essentially a specialization of \cref{alg:blocking-flow-union} to the case where all the $k$ matroids are the same, except that we skip exploring the out-neighbors of $a_{\ell}$ and remove it directly if it is ``spanned'' by the previous layers and the set $R_{\ell} \subseteq L_{\ell}$ of elements that are \emph{not} on any augmenting path of length $d_t$.
With this optimization, we obtain the following lemma analogous to \cref{lemma:blocking-flow-union}.

\begin{algorithm}
  \SetEndCharOfAlgoLine{}
  \SetKwInput{KwData}{Input}
  \SetKwInput{KwResult}{Output}
  \SetKwInput{KwGuarantee}{Guarantee}
  
  \caption{Blocking flow in a $k$-fold union exchange graph}
  \label{alg:blocking-flow-union-fold}
  \KwData{$S \subseteq U$ which partitions into $S_1, \ldots, S_k$ of independent sets and a dynamic-basis data structure $\cD$ of $U \setminus S$}
  \KwResult{$S^\prime \in \cI_{\text{part}} \cap \cI_k$ with $d_{H(S^\prime)}(s, t) > d_{H(S)}(s, t)$}
  \KwGuarantee{$\cD$ maintains a basis of $U \setminus S^\prime$ at the end of the algorithm}
  
  Build the distance layers $L_2, \ldots, L_{d_t - 1}$ of $H(S)$ with \cref{lemma:bfs-union-fold}\;
  $L_0 \gets \{s\}$ and $L_{d_t} \gets \{t\}$\;
  $B \gets$ the basis maintained by $\cD$ and $L_1 \gets B$\;
  $A_\ell \gets L_\ell$ for each $0 \leq \ell \leq d_t$\;
  $\cT_{\ell}^{(i)} \gets \textsc{Initialize}(\cM_i, S_i, Q_{S_i}, A_{\ell} \cap S_i)$ for each $2 \leq \ell < d_t$ and $1 \leq i \leq k$ (\cref{thm:bst} with $\beta = \sqrt{r} / d_t$)\;
  $D_{\ell} \gets \emptyset$ for each $1 \leq \ell < d_t$\;
  $R_{\ell} \gets \emptyset$ for each $2 \leq \ell < d_t$\;
  $\ell \gets 0$ and $a_0 \gets s$\;
  \While{$\ell \geq 0$} {
    \If{$\ell < d_t$} {
      \lIf{$A_{\ell} = \emptyset$} {
        \textbf{break}
      }
      \lIf{$\ell \geq 2$ and $\rank(L_1 \cup \cdots \cup L_{\ell - 1} \cup R_{\ell} \cup \{a_{\ell}\}) = \rank(L_1 \cup \cdots \cup L_{\ell - 1} \cup R_{\ell})$} {\label{line:check-not-spanned}
        $A_{\ell} \gets A_{\ell} - a_{\ell}$ and \textbf{continue}
      }
      \lIf{$\ell > 0$} {
        Find an $a_{\ell + 1} := \cT_{\ell + 1}^{(i)}.\textsc{Find}(a_\ell) \neq \bot$ for some $1 \leq i \leq k$
      }
      \lElse {
        $a_{\ell + 1} \gets$ an arbitrary element in $A_1$
      }
      \If{such an $a_{\ell + 1}$ does not exist} {
        \lIf{$\ell \geq 2$} {
          $R_{\ell} \gets R_{\ell} + a_{\ell}$ and $\cT_{\ell}^{(j)}.\textsc{Delete}(a_{\ell})$ where $a_{\ell} \in S_j$
        }
        $A_\ell \gets A_\ell - a_\ell$ and $\ell \gets \ell - 1$\;\label{line:remove}
      }
      \lElse {
        $\ell \gets \ell + 1$
      }
    }
    \Else {
    \tcp{Found augmenting path $a_1, a_2, \ldots a_\ell$}
      $B \gets B - a_1$, $A_1 \gets A_1 - a_1$, and $D_1 \gets D_1 + a_1$\;
      \If{$\cD.\textsc{Delete}(a_1)$ returns a replacement $x$} {
        $B_i \gets B_i + x$  and $A_i \gets A_i + x$\;
      }
      \For{$i \in \{2, \ldots, d_t - 1\}$} {
          $D_{i} \gets D_{i} + a_{i}$ and $A_{i} \gets A_{i} - a_{i}$\;
          $\cT_{i}^{(j)}.\textsc{Delete}(a_i)$ and $\cT_{i}^{(j)}.\textsc{Update}(\{a_{i - 1}, a_{i}\})$ where $a_i \in S_j$\;
      }
      Augment $S$ along $P = (s, a_1, \ldots, a_{d_t - 1}, t)$\;
      $\ell \gets 0$\;
    }
  }
  \textbf{return} $S$\;
\end{algorithm}

\begin{lemma}
  Given an $S \in \cI_{\text{part}} \cap \hat{\cI}$ with $d_{H(S)}(s, t) = d_t$ together with a data structure $\cD$ of \cref{thm:decremental-basis} that maintains a basis of $U \setminus S$, it takes
  \begin{equation}
    \tO\left(\underbrace{kr}_{\ref{item:term1}} + \underbrace{\left(|S^\prime| - |S|\right) \cdot d_t\sqrt{r}}_{\ref{item:term2}} + \underbrace{\left((|S^\prime| - |S|) \cdot d_t + r\right) \cdot k}_{\ref{item:term3}} + \underbrace{\left(kr + (|S^\prime| - |S|)\right) \cdot \frac{\sqrt{r}}{d_t}}_{\ref{item:term4}}\right)
    \label{eq:complexity}
  \end{equation}
  time to obtain an $S^\prime \in \cI_{\text{part}} \cap \hat{\cI}$ with $d_{H}(S^\prime)(s, t) > d_t$, with an additional guarantee that $\cD$ now maintains a basis of $U \setminus S^\prime$.
  \label{lemma:blocking-flow-union-fold}
\end{lemma}

We need the following observation to bound the running time of \cref{alg:blocking-flow-union-fold}.

\begin{observation}
  In \cref{alg:blocking-flow-union-fold}, it holds that $B \cup R_2 \cup R_{3} \cup \cdots \cup R_{d_t - 1} \in \cI$.
  \label{claim:is-basis}
\end{observation}

\begin{proof}[Proof of \cref{lemma:blocking-flow-union-fold}]
  We analyze the running time of \cref{alg:blocking-flow-union-fold} first.
  In particular, there are four terms in \cref{eq:complexity} which come from the following.
  \begin{enumerate}[label=(\roman*)]
    \item\label{item:term1} $\tO(kr)$: It takes $\tO(kr)$ time to compute the distance layers using \cref{lemma:bfs-union-fold} and initialize all the binary search trees $\cT_{\ell}^{(i)}$'s. Computing the rank of $L_1 \cup \cdots \cup L_{\ell - 1} \cup R_{\ell}$ also takes $\tO(kr)$ time in total since we can pre-compute query-sets of the form $L_1 \cup \cdots \cup L_{k}$ for each $k$ in $\tO(kr)$ time, and each insertion to $R_{\ell}$ takes $\tO(1)$ time.
    \item\label{item:term2} $\tO\left(\left(|S^\prime| - |S|\right) \cdot d_t\sqrt{r}\right)$: For each of the $O(|S^\prime| - |S|)$ augmentations, it takes $\tO(r \cdot \frac{d_t}{\sqrt{r}})$ time to update the binary search trees.
    \item\label{item:term3} $\tO(\left((|S^\prime| - |S|) \cdot d_t + r\right) \cdot k)$: The number of elements whose out-edges are explored is bounded by $O\left((|S^\prime| - |S|) \cdot d_t + r\right)$. This is because for each such element $u$, either $u$ is included in an augmenting path of length $d_t$, or $u$ is removed in Line \ref{line:remove}.
    There are $O((|S^\prime| - |S|) \cdot d_t)$ such $u$'s in the augmenting paths.
    For $u$ removed in Line \ref{line:remove}, if $\ell = 1$, then the number of such $u$'s is $O(|S^\prime| - |S| + r)$ because there are initially $O(r)$ elements in $A_1$, and we add at most one to it every augmentation.
    If $\ell \geq 2$, then we insert it into $R_{\ell}$, and by Line \ref{line:check-not-spanned}, the rank of $L_1 \cup \cdots \cup L_{\ell - 1} \cup R_{\ell}$ increases after including $u$ into $R_{\ell}$.
    By \cref{claim:is-basis}, the number of such $u$'s is bounded by $O(r)$.
    The term then comes from spending $O(k)$ time iterating through the $k$ binary search trees for each of the $O\left((|S^\prime| - |S|) \cdot d_t + r\right)$ elements whose out-neighbors are explored.
    \item\label{item:term4} $\tO(\left(kr + (|S^\prime| - |S|)\right) \cdot \frac{\sqrt{r}}{d_t})$:  The number of elements that are once in some $A_{\ell}$ is bounded by $O(kr + |S^\prime| - |S|)$. Initially, there are $O(kr)$ elements ($A_1$ plus all the $A_\ell$ for $\ell \geq 2$), and each augmentation adds at most one element to $A_1$.
    Each of these elements is discovered by $\cT_{\ell}^{(i)}.\textsc{Find}(\cdot)$ at most once, and thus we can charge the $\tO(\frac{\sqrt{r}}{d_t})$ cost to it, resulting in the fourth term of \cref{eq:complexity}.
  \end{enumerate}
  Note that for each element whose out-neighbors are explored, any failed attempt of $\cT_{\ell}^{(i)}.\textsc{Find}(\cdot)$ costs only $\tO(1)$ instead of $\tO(\frac{\sqrt{r}}{d_t})$ according to \cref{thm:bst}.
  The $\tO(\frac{\sqrt{r}}{d_t})$ cost of a successful search is charged to term \ref{item:term4} instead of \ref{item:term3}.
  
  As for correctness, it suffices to show that each of the $a_{\ell}$ removed from $A_{\ell}$ because it is spanned by $L_1 \cup \cdots \cup L_{\ell - 1} \cup R_{\ell}$ in Line \ref{line:check-not-spanned} is not in any augmenting path of length $d_t$.
  Consider its out-neighbor $a_{\ell + 1}$ with respect to the current $S$, and we would like to argue that $a_{\ell + 1}$ is not on any augmenting path of length $d_t$ anymore.
  This is because we have already explored all the out-neighbors of elements in $R_{\ell}$.
  Since $a_{\ell} \in \sspan(L_1 \cup \cdots \cup L_{\ell - 1} \cup R_{\ell})$, by \cref{lemma:basis-rank}, there must exist some $u \in L_1 \cup \cdots \cup L_{\ell - 1} \cup R_{\ell}$ with a directed edge $(u, a_{\ell + 1})$.
  We consider two cases:
  \begin{itemize}
  \item $u\in R_{\ell}$. This means that we have already explored $a_{\ell+1}$, as we finished exploring all out-neighbors of $u$ already.
  \item $u\in L_{1}\cup \cdots \cup L_{\ell-1}$. We know that by \cref{lemma:monotone}, both $d_{H(S)}(s,v)$ and
  $d_{H(S)}(v,t)$ can only increase after augmentations for all elements $v$.
  Hence $a_{\ell+1}$ cannot be part of an augmenting path of length $d_t$ anymore, since if it was its distance to $t$ must be $d-(\ell+1)$, but then the distance from $u$ to $t$ must be at most $d-\ell$ (which is smaller than its initial distance to $t$ at the beginning of the phase).
  \end{itemize}
  
  As a result, all of $u$'s out-neighbors have either already been explored or do not belong to any augmenting path of length $d_t$.
  This implies that $u$ is not on any such path either, and thus it's correct to skip and remove it from $A_{\ell}$.
  This concludes the proof of \cref{lemma:blocking-flow-union-fold}.
\end{proof}

\cref{thm:dynamic-matroid-union-fold} now follows from analyzing the total running time of $O(\sqrt{\min(n, kr)})$ runs of \cref{lemma:blocking-flow-union-fold}.

\begin{proof}[Proof of \cref{thm:dynamic-matroid-union-fold}]
  We initialize the dynamic-basis data structure $\cD$ of \cref{thm:decremental-basis} on $U$ in $\tO(n)$ time.
  Let $p = \min(n, kr)$ be the rank of the $k$-fold matroid union.
  Using $\cD$, we then run $O(\sqrt{p})$ iterations of \cref{lemma:blocking-flow-union-fold} until $d_{H(S)}(s, t) > \sqrt{p}$.
  Summing the first two terms of \cref{eq:complexity} over these $O(\sqrt{p})$ iterations gives (recall that \cref{lemma:augmenting-path-lengths} guarantees that $\sum_{d = 1}^{\sqrt{p}}d \cdot (|S_d| - |S_{d - 1}|) = \tO(p)$)
  \[
    \tO\left(kr\sqrt{p} + \sqrt{r} \cdot \sum_{d = 1}^{\sqrt{p}}d \cdot \left(|S_d| - |S_{d - 1}|\right)\right) = \tO\left(kr\sqrt{p}\right)
  \]
  since $p\sqrt{r} \leq kr\sqrt{p}$.
  The third term of \cref{eq:complexity} contributes a total running time of
  \[
    \tO\left(\left(\sum_{d = 1}^{\sqrt{p}}dk \cdot \left(|S_d| - |S_{d - 1}|\right)\right) + kr\sqrt{p}\right) = \tO\left(kr\sqrt{p} + kp\right),
  \]
  while the fourth term of \cref{eq:complexity} sums up to
  \[
    \tO\left(\left(\sum_{d = 1}^{\sqrt{p}}kr\frac{\sqrt{r}}{d}\right) + kr\sqrt{r}\right) = \tO\left(kr\sqrt{r}\right).
  \]
  We finish the algorithm by finding the remaining $O(\sqrt{p})$ augmenting paths one at a time with \cref{lemma:bfs-union-fold} in a total of $\tO(kr\sqrt{p})$ time.
  The $k$-fold matroid union algorithm thus indeed runs in $\tO\left(n + kr\sqrt{\min(n, kr)} + k\min(n, kr)\right)$ time, concluding the proof of \cref{thm:dynamic-matroid-union-fold}.
\end{proof}

\section{Dynamic Oracles for Specific Matroids \& Applications} \label{appendix:applications}

In this appendix, we show how to leverage known dynamic algorithms to implement the dynamic rank oracle (\cref{def:dyn-oracle}) efficiently for many important matroids.
What we need are data structures that can maintain the rank of a set dynamically under insertions and deletions in \emph{worst-case} update time (converting a \emph{worst-case} data structure to \emph{fully-persistent} can be done by the standard technique of \cite{DriscollSST86,Dietz89}, paying an overhead of $O(\log{n})$). Additionally, note that the data structures do not need to work against an adaptive adversary since we only ever use the \emph{rank} of the queried sets, which is not affected by internal randomness.

In particular, for \emph{partition, graphic, bicircular, convex transversal}, and \emph{simple job scheduling matroids} it is possible to maintain the rank with polylog$(n)$ update-time, and for \emph{linear matroids} in $O(n^{1.529})$ update-time.

Together with our matroid intersection (\cref{sec:matroid-intersection}) and matroid union (\cref{sec:matroid-union}) algorithms, this leads to a black-box approach to solving many different problems. In fact, we can solve matroid intersection and union on any combination of the above matroids, leading to improved or matching running times for many problems (see the introduction \cref{sec:intro} with \cref{tab:intro:dyn-oracle-implies-fast-algorithms,tab:intro:implications} for a more thorough discussion). For completeness, we define these problems in \cref{appendix:problems}.
The same algorithms are powerful enough to also solve new problems which have not been studied before.

\paragraph{Example Application: Tree Scheduling (or Maximum Forest with Deadlines).} We give an example of a reasonably natural combinatorial scheduling problem, which---to our knowledge---has not been studied before. Suppose we are given a graph $G = (V,E)$ where each edge $e\in E$ has two numbers associated with it: a release date $\ell_e$ and a deadline $r_e$. Consider the problem where we want to for each day pick exactly one edge (say, to build/construct), but we have constraints that edge $e$ can only be picked between days $\ell_e$ and $r_e$. Now the task is to come up with a scheduling plan to build a spanning tree of the graph, if possible.

This problem is exactly a matroid intersection problem between a graphic matroid and a convex transversal matroid. Hence, by a black-box reduction, we know that we can solve this problem in $\tO(|E|\sqrt{|V|})$ time.

\subsection{Partition Matroids}
In a partition matroid $\cM = (\ground,\cI)$, each element $u\in \ground$ is assigned a color $c_u$. We are also, for each color $c$, given a non-negative integer $d_c$, and we define a set of elements $S\subseteq \ground$ to be independent if for each color $c$, $S$ includes at most $d_c$ elements of this color.
Implementing the dynamic oracle for the partition matroid is easy:

\begin{lemma}
One can maintain the rank of a partition matroid in $O(1)$-update time.
\end{lemma}
\begin{proof}
For each color $c$ we maintain a counter $x_c$ of how many elements we have of color $c$. We also maintain $r = \sum \min_c(x_c, d_c)$, which is the rank of the current set.
\end{proof}

\begin{remark}
Bipartite matching can be modeled as a matroid intersection problem of two partition matroids. So our matroid intersection algorithm together with the above lemma match (up to poly-logarithmic factors induced by fully-persistence)---in a black-box fashion---the $O(|E|\sqrt{|V|})$-time bound of the best \emph{combinatorial} algorithm for bipartite matching \cite{HopcroftK73}.
\end{remark}

\subsection{Graphic and Bicircular Matroids}
\label{appendix:applications-graphic}
Given a graph $G = (V,E)$, the \emph{graphic} and \emph{bicircular} matroids are matroids capturing the connectivity structure of the graph.

\paragraph{Graphic Matroid.}
In the graphic matroids $\cM = (E,\cI)$ a subset of edges $E'\subseteq E$ are independent if and only if they do not contain a cycle.
We use the following result to implement the dynamic oracle for this matroid.

\begin{lemma}[\cite{KapronKM13,GibbKKT15}]
  There is a data structure that maintains an initially-empty graph $G = (V, E)$ and supports insertion/deletion of edges $e$ into/from $E$ in worst case $O(\log^4{|V|})$ time and query of the connectivity between $u$ and $v$ in worst case $O(\log{|V|}/\log{\log{|V|}})$ time.
  The data structure works with high probability against an oblivious adversary.
  \label{lemma:worst-case-dynamic-conn-appendix}
\end{lemma}

With a simple and standard extension, we can maintain the number of connected components as well, and hence also the rank (since $\rank(E') = |V|- \#\text{connected components in }G[E']$).

\begin{corollary}
  There is a data structure that maintains an initially-empty graph $G = (V, E)$ and supports insertion/deletion of $e$ into/from $E$ in worst-case $O(\log^4{|V|})$ time.
  After each operation, the data structure also returns the number of connected components in $G$.
  The data structure works with high probability against an oblivious adversary.
  \label{cor:worst-case-dynamic-components-appendix}
\end{corollary}

\begin{proof}
  We maintain the data structure $\cC$ of \cref{lemma:worst-case-dynamic-conn-appendix} and a counter $c := |V|$ representing the number of connected components.
  For insertion of $e = (u, v)$, we first query the connectivity of $u$ and $v$ before inserting $e$ into $\cC$.
  If they are not connected before the insertion, decrease $c$ by one.
  For deletion of $e = (u, v)$, after deleting $e$ from $\cC$, we check if $u$ and $v$ are still connected.
  If not, then we increase $c$ by one.
\end{proof}

\paragraph{Bicircular Matroid.}
In the bicircular matroid $\cM = (E,\cI)$, a subset of edges $E'\subseteq E$ are independent if and only if each connected component in $G[E']$ has at most one cycle. Similar to the graphic matroid, dynamic connectivity algorithms can be used to implement the dynamic rank oracle for bicircular matroids too.

\begin{corollary}
  There is a data structure that maintains an initially-empty graph $G = (V, E)$ and supports insertion/deletion of $e$ into/from $E$ in worst-case $O(\log^4{|V|})$ time.
  After each operation, the data structure also returns the rank of $E$ in the bicircular matroid.
  The data structure works with high probability against an oblivious adversary.
  \label{cor:worst-case-dynamic-components-appendix-bicircular}
\end{corollary}

\begin{proof}
The dynamic connectivity data structure of 
\cite{KapronKM13,GibbKKT15} (\cref{lemma:worst-case-dynamic-conn-appendix}) can be adapted to also keep track of the number of edges and vertices in each connected component. Using this, the data structure can, for each connected component $c$ keep track of a number $x_c$ as the minimum of the number of edges in this component and the number of vertices in this component. Then the rank of the bicircular matroid is just the sum of $x_c$ (as in an independent set each component is either a tree or a tree with an extra edge). In each update two components can merge, a component can be split up into two, or the edge-count of a component may simply change. 
\end{proof}

\begin{remark}[Deterministic Dynamic Connectivity]
The above dynamic connectivity data structures are randomized.
There are also deterministic connectivity data structures, but with slightly less efficient sub-polynomial $|V|^{o(1)}$ update time~\cite{ChuzhoyGLNPS20}.
\end{remark}

\subsection{Convex Transversal and Scheduling Matroids}

Convex transversal and  scheduling matroids are special cases of the \emph{transversal} matroid, with applications in scheduling algorithms.

\begin{definition}[Transversal Matroid~\cite{edmonds1965transversals}]
  A \emph{transversal matroid} with respect to a bipartite graph $G = (L, R, E)$ is defined over the ground set $L$, where each $S \subseteq L$ is independent if and only if there is a perfect matching in $G$ between $S$ and a subset of $R$.
\end{definition}

A bipartite graph $G = (L, R, E)$ is \emph{convex} if $R$ has a linear order $R = \{r_1, r_2, \ldots, r_n\}$ and each $\ell \in L$ corresponds to an interval $1 \leq s(\ell) \leq t(\ell) \leq n$ such that $(\ell, r_i) \in E$ if and only if $s(\ell) \leq i \leq t(\ell)$, i.e., the neighbors of each $\ell$ form an interval.

\begin{definition}[Convex Transversal Matroid and Simple Job Scheduling Matroid]
  A \emph{convex transversal matroid} is a transversal matroid with respect to a convex bipartite graph.
  A \emph{simple job scheduling matroid} is a special case of convex transversal matroids in which $s(\ell) = 1$ for each $\ell \in L$.
\end{definition}

One intuitive way to think about the simple job scheduling matroid is that there is a machine capable of finishing one job per day.
The ground set of the matroid consists of $n$ jobs, where the $i$-th jobs must be done before its deadline $d_i$.
A subset of jobs forms an independent set if it's possible to schedule these jobs on the machine so that every job is finished before its deadline.

\begin{lemma}[Dynamic Convex Bipartite Matching \cite{BrodalGHK07}]
  There is a data structure which, given a convex bipartite graph $G = (L, R, E)$, after $\tO(|L|+|R|)$ initialization, maintains the size of the maximum matching of $G[A \cup R]$ where $A \subseteq L$ is a dynamically changing subset of $L$ that is initially empty.
  The data structure supports insertion/deletion of an $x \in L$ to/from $A$ in worst-case $O(\log^{2}(|L|+|R|))$ update time.
\end{lemma}

\begin{remark}
  The exact data structure presented in \cite{BrodalGHK07} is different from the stated one.
  In particular, they support insertion/deletion of an \emph{unknown} job, i.e., we do not know beforehand what the starting date and deadline of the job are, nor do we know its relative position among the current set of jobs.
  As a result, they used a rebalancing-based or rebuilding-based binary search tree~\cite{NievergeltR72,Andersson89,Andersson91}, resulting in their amortized bound.
  For our use case, all the possible jobs are known and we are just activating/deactivating them, hence a static binary tree with a worst-case guarantee over these jobs suffices.
\end{remark}

\subsection{Linear Matroid}
In a linear matroid $\cM = (\ground, \cI)$, $\ground$ is a set of $n$ vectors (of dimension $r$) in some vector space and the notion of independence is just that of \emph{linear independence}. The dynamic algorithm to maintain the rank of a matrix of \cite{BrandNS19} can be used without modification as the dynamic oracle.

\begin{lemma}[Dynamic Matrix Rank Maintenance~\cite{BrandNS19}]
  There is a data structure which, given an $n \times n$ matrix $M$, maintains the rank of $M$ under row updates in worst-case $O(n^{1.529})$ update time.
\end{lemma}

\subsection{Problems} \label{appendix:problems}
For completeness, here we define the problems we discuss in the introduction, and why they reduce to matroid union or intersection.

\paragraph{$k$-Forest.} In this problem we are given a graph $G = (V,E)$ and asked to find $k$ edge-disjoint forests of the graph, of the maximum total size.
It can be modeled as the $k$-fold matroids union over the graphic matroid of $G$.

\paragraph{$k$-Disjoint Spanning Trees.} This problem is a special case of the above $k$-forest problem where we ask to find $k$ edge-disjoint spanning trees of the graph. Clearly, if such exists, the $k$-forest problem will find them.

\paragraph{$k$-Pseudoforest.} Similar to above, in this problem we are given a graph $G = (V,E)$ and asked to find $k$ edge-disjoint \emph{pseudoforests} of the graph, of the maximum total size. A pseudoforest is an undirected graph in which every component has at most one cycle.
The problem can be modeled as the $k$-fold matroids union over the bicircular matroid of $G$.

\paragraph{$(f,p)$-Mixed Forest-Pseudoforest.} Again, we are given a graph $G = (V,E)$ and asked to find $f$ forests and $p$ pseudoforest (all edge-disjoint), of the maximum total size. 
The problem can be modeled as the matroids union over $f$ graphic matroids and $p$ bicircular matroids.

\paragraph{Tree Packing.} In the tree packing problem, we are given a graph $G = (V,E)$ and are asked to find the maximum $k$ such that we can find $k$-disjoint spanning trees in the graph. This number $k$ is sometimes called the \emph{tree-pack-number} or \emph{strength} of the graph. The problem can be solved with the $k$-disjoint spanning trees problem, by binary searching for $k$ in the range $[0,|E|/(|V|-1)]$, and is an example of a \emph{matroid packing} problem.

\paragraph{Arboricity and Pseudoarboricity.} The arboricity (respectively pseudoarboricity) of a graph $G = (V,E)$ is the least integer $k$ such that we can partition the edges into $k$ edge-disjoint forests (respectively pseudoforests). 
This can be solved with the $k$-forest (respectively $k$-pseudoforest) problem with a binary search over $k$. It is well known that for a simple graph the (pseudo-)arboricity is at most $\sqrt{|E|}$, so we need only search for $k$ in the range $[0,\sqrt{|E|}]$. The problems are examples of \emph{matroid covering} problems.

\paragraph{Shannon Switching Game.} The Shannon switching game is a game played on a graph $G = (V,E)$, between two players ``Short'' and ``Cut''. They alternate turns with Short playing first, and all edges are initially colored white. On Short's turn, he may color an edge of the graph black. On Cut's turn, he picks a remaining non-black edge and removes it from the graph. Short wins if he connects the full graph with black edges, and Cut if he manages to disconnect the graph. It can be shown that Short wins if and only if there exists two disjoint spanning trees in the graph (and these two spanning trees describes a winning strategy for Short). Hence solving this game is a special case of the $k$-disjoint spanning tree problem with $k = 2$.

\paragraph{Graph $k$-Irreducibility.} 
A (multi-)graph $G = (V,E)$ is called $k$-irreducible (\cite{Whiteley88rigidity}) if and only if $|E| = k(|V|-1)$ and for any vertex-induced nonempty, proper subgraph $G[V']$ it holds that $|E(G[V'])| < k(|V'|-1)$.
The motivation behind this definition comes from the rigidity of \emph{bar-and-body frameworks}. A bar-and-body framework where rigid bars are attached to rigid bodied with joints (represented by the graph $G$).
Then any stress put on a $k$-irreducible structure will propagate to all the bars (i.e.\ edges).
\cite{gabow1988forests} show how one can decide if a graph is $k$-irreducible by first determining if its edges can be partitioned into $k$ edge-disjoint trees, and then performing an additional $\tO(k|V|)$ work.

\paragraph{Bipartite Matching.} In the bipartite matching problem, we are given a bipartite graph $G = (L\cup R,E)$, and the goal is to find a \emph{matching} (a set of edges which share no vertices) of maximum size. Bipartite matching can be modeled as a matroid intersection problem over two partition matroids $M_L = (E,\cI_L)$ and $M_R = (E,\cI_R)$. $M_L$ specifies that no two edges share the same vertex on the left $L$ (and $M_R$ is defined similarly on the right  set of vertices $R$).

\paragraph{Colorful Spanning Tree.} In this problem\footnote{sometimes also called \emph{rainbow spanning tree}.}, we are given a graph $G = (V,E)$ together with colors on the edges $c:E \to \bbZ$. We are tasked to find a spanning tree of $G$ such that no two edges in our spanning tree have the same color. This problem can be modeled by the matroid intersection of the graphic matroid of $G$ (ensuring we pick a forest),
and a partition matroid of the coloring $c$ (ensuring that we pick no duplicate colors). We also note that this problem is more difficult than bipartite matching since any bipartite matching instance can be converted to a colorful spanning tree instance on a star-multi-graph.

\paragraph{Graphic Matroid Intersection.} In graphic matroid intersection we are given two graphs $G_1 = (V_1,E_1)$ and $G_2 = (V_2,E_2)$ and a bijection of the edges $\phi : E_1 \to E_2$. The task is to find a forest in $G_1$ of the maximum size, which also maps to a forest in $G_2$. By definition, this is a matroid intersection problem over two graphic matroids. Again, this problem is a further generalization of the colorful spanning tree problem.

\paragraph{Convex Transversal and Simple Job Scheduling Matroid Intersection.} In these problems, we are given a set of unit-size jobs $V$, where each job $v$ has two release times $\ell_1(v)$, $\ell_2(v) \ge 1$ (in simple job scheduling $\ell(v) = 1$) and two deadlines $r_1(v)$,$r_2(v)\le \mu$. The task is to find a set of jobs $S$ of the maximum size such that they can be scheduled on two machines as follows: each job needs to be scheduled at both machines, and at machine $i$ it must be scheduled at time $t \in [\ell_i(v), r_i(v)]$.

\paragraph{Linear Matroid Intersection}. In this problem, we are given two $n\times r$ matrices $M_1$ and $M_2$ over some field.
The task is to find a set of indices $S\subseteq \{1,2,\ldots,n\}$ of maximum cardinality, such that the rows of $M_1$ (respectively $M_2$) indexed by $S$ are independent at the same time. This is a matroid intersection of two linear matroids defined by $M_1$ and $M_2$. We note that partition, graphic, and transversal matroids are special cases of linear matroids.

\section{Independence-Query Matroid Intersection Algorithm} \label{appendix:dynamic-matroid-intersection-ind}

In this section, we show that we can obtain an $\tO(nr^{3/4})$ matroid intersection algorithm in the dynamic-independence-oracle model.
This matches the state-of-the-art traditional independence-query algorithm of Blikstad \cite{blikstad21}.
We will only provide a proof sketch here because our algorithm is mostly an implementation of \cite{blikstad21} in the new model with the help of (circuit) binary search trees.

Using the same construction as \cref{thm:bst} and \cref{obs:exchange}, circuit binary search trees work verbatim in the dynamic-independence-oracle model (however, co-circuit binary search trees do not).
In particular, \cref{obs:exchange}\ref{item:circuit-exchange} can be checked with a single independence query.

\begin{corollary}
  For any integer $\beta \geq 1$, there exists a data structure that supports the following operations.
  \begin{itemize}
    \item $\textsc{Initialize}(\cM, S, Q_S, X)$: Given $S \in \cI$, the query-set $Q_S$ that corresponds to $S$, and $X \subseteq S$ or $X = \{t\}$, initialize the data structure in $\tO(|X|)$ time. The data structure also maintains $S$.
    \item $\textsc{Find}(y)$: Given $y \in \bar{S}$,
      \begin{itemize}
        \item if $X \subseteq S$, then return an $x \in X$ such that $S - y + x \in \cI$, or
        \item if $X = \{t\}$, then return the only element $x = s$ or $x = t$ in $X$ if $S + y \in \cI$ and $\bot$ otherwise.
      \end{itemize}
      The procedure returns $\bot$ if such an $x$ does not exist.
      The procedure takes $\tO(\beta)$ time if the result is not $\bot$, and $\tO(1)$ time otherwise.
    \item $\textsc{Delete}(x)$: Given $x \in X$, if $x \not\in \{s, t\}$, delete $x$ from $X$ in $O(\log{n})$ time.
    \item $\textsc{Replace}(x, y)$: Given $x \in X$ and $y \not\in X$, replace $x$ in $X$ by $y$ in $O(\log{n})$ time.
    \item $\textsc{Update}(\Delta)$: Update $S$ to $S \oplus (\Delta \setminus \{s, t\})$ in amortized $\tO(\frac{|X| \cdot |\Delta|}{\beta})$ time.
  \end{itemize}
  \label{cor:bst-ind}
\end{corollary}

\paragraph{Framework.}
The algorithm of \cite{blikstad21} consists of the following three phases.

\begin{enumerate}
  \item First, obtain an $(1 - \epsilon)$-approximate solution $S$ using augmenting sets in $\tO(\frac{n\sqrt{r}}{\epsilon})$ time.
  \item Eliminate all augmenting paths in $G(S)$ of length at most $d$ using Cunningham's algorithm implemented by \cite{chakrabarty2019faster} in $\tO(nd + nr\epsilon)$ time.
  \item Finding the remaining $O(r/d)$ augmenting paths one at a time, using $\tO(n\sqrt{r})$ time each.
\end{enumerate}

With $\epsilon = r^{-1/4}$ and $d = r^{3/4}$, the total running time is $\tO(nr^{3/4})$.
We briefly sketch how to implement the above three steps in the same running time also for the dynamic-independence-oracle model.

Note that the primary difficulty independence-query algorithms face is that we are only capable of checking \cref{obs:exchange}\ref{item:circuit-exchange} (using \cref{cor:bst-ind}), which means that we can only explore the neighbors of $u \in \bar{S}$.
The aforementioned rank-query algorithms for building distance layers and finding blocking-flow style augmenting paths are thus inapplicable in the dynamic-independence-oracle model.

\paragraph{Approximation Algorithm.}
The $O(n\sqrt{r}/\epsilon)$-query $(1-\epsilon)$-approximation algorithm of \cite{blikstad21} needs to first compute distance layers up to distance $O(\frac{1}{\epsilon})$. This is done in a similar way as sketched below for ``Eliminating Short Augmenting Paths''.

Otherwise, the approximation algorithm works through a series of ``refine'' operations
(algorithms \texttt{RefineAB}, \texttt{RefineBA}, and \texttt{RefineABA} in \cite{chakrabarty2019faster,blikstad21}) to build a partial augmenting set. In each such operation, we only need to be able to do the following for some sets $(P,Q)$: start from some set $Q$ and find a maximal set $X\subseteq P$ such that
$Q+X$ is independent. This can be performed with a greedy algorithm in time (and dynamic query) $O(|P|)$, given that we already have queried set $Q$ before (which will be the case). 

Finally, the approximation algorithm falls back to finding a special type of augmenting paths \emph{with respect to} the current augmenting set, in the \texttt{RefinePath} algorithm of \cite{blikstad21}, with $\tO(n)$ queries for each such path. This algorithm can be implemented also in the dynamic-oracle model with the same query complexity. \texttt{RefinePath} relies on the \texttt{RefineAB} and \texttt{RefineBA} algorithms (which we already covered), in addition to a binary search trick to find feasible exchange pairs. This binary search can be implemented with the circuit trees (\cref{cor:bst-ind}), and it takes a total time $\tO(n)$ to build them (since we can keep track of a queried set for $S$, and then we only need to build a circuit tree statically once for each layer in cost proportional to the size of the layer---which sums up to $\tO(n)$).

\paragraph{Eliminating Short Augmenting Paths.}
Using \cite{chakrabarty2019faster}'s implementation of Cunningham's algorithm, we can eliminate all augmenting paths of length $d$, thereby obtaining a $(1 - 1/d)$-approximate solution.
The algorithm relies on \cref{lemma:monotone} to ``fix'' the distance layers after each augmentation.
Initially, all elements have distance $1$ or $2$ from $s$ depending on whether it belongs to $S$ (the common independent set obtained by the above approximation algorithm) or not.
Before the first and after each of the remaining $O(\epsilon r)$ augmentation, we can fix the distance layers as follows.

\begin{itemize}
  \item For each $1 \leq \ell \leq d$ and $u \in L_{\ell}$, if $u$ is not of distance $\ell$ from $s$, i.e., there is no in-edge from $L_{\ell - 1}$ to $u$ anymore, move $u$ from $L_{\ell}$ to $L_{\ell + 2}$. This check is done as follows, depending on the parity of $\ell$.
  \begin{itemize}
    \item If $\ell$ is even, then for each $v \in L_{\ell - 1}$, we find all the unmarked $u \in L_{\ell}$ that $v$ has an edge to and mark $u$.
    In the end, all the unmarked $u \in L_{\ell}$ do not belong to $L_{\ell}$ and should be moved to $L_{\ell + 2}$.
    \item If $\ell$ is odd, then we simply check if there is an in-edge from $L_{\ell - 1}$ to decide whether $u$ should be moved to other layers.
  \end{itemize}
\end{itemize}

Both cases can be implemented efficiently with circuit binary search trees of \cref{cor:bst-ind}: Each time we spend $\tO(1)$ time to either confirm that $u$ has distance $\ell$ from $s$ with respect to the current $S$ (in which case it will not be moved anymore in this iteration), or we increase the distance estimate of $u$.
The total running time is thus $\tO(nd + nr\epsilon)$, where $\tO(nd)$ comes from increasing the distance estimate of each element to at most $d$, and $\tO(nr\epsilon)$ comes from confirming that each element belongs to its distance layer in the $O(\epsilon r)$ iterations.

A caveat here is that we need to support insertion/deletion into the binary search trees.
This can be made efficient by doubling the size of a binary search tree (and re-initializing it) every time when there are not enough leaf nodes left in it.
The cost of re-building will be amortized to $\tO(1)$ time per update (i.e., movement of an element to another layer).

\paragraph{Finding a Single Augmenting Path.}
With the $(1 - 1/d)$-approximate solution obtained in the first two steps, \cite{blikstad21} then finds the remaining $O(r/d)$ augmenting paths one at a time, using the \emph{reachability} algorithm of \cite{quadratic2021}.
The reachability algorithm roughly goes as follows.
First, we initialize two circuit binary search trees (\cref{cor:bst-ind}) over the two matroids for discovering out-edges and in-edges of elements in $\bar{S}$.
We then repeatedly run the following three steps until either an $(s, t)$-path is found (an arbitrary $(s, t)$-path suffices since such a path can be converted into a chordless one in $\tO(r)$ time along which augmentation is valid) or we conclude that $t$ is unreachable from $s$.
We keep track of a set of visited vertices $F$ which we know are reachable from $s$.

\begin{enumerate}[label=(\roman*)]
  \item Identify the set of unvisited \emph{heavy} vertices in $\bar{S}$ that have at least $\sqrt{r}$ unvisited out-neighbors or has a direct edge toward $t$. This is done by sampling a set $R$ of unvisited vertices in $S$ and then computing for each vertex $u$ whether $R \cap \mathsf{OutNgh}(u) = \emptyset$, or equivalently, whether $S - R + u \in \cI$.
  Intuitively, vertices with more out-neighbors are more likely to fail the test.
  This can be tested for a single $R$ and all $u$ in $O(n)$ time in the dynamic-independence-oracle model.
  With $O(\log{n})$ samples, heavy vertices can be successfully identified with high probability.
  \item Discover all the out-neighbors for each light vertex, taking a total of $\tO(n\sqrt{r})$ time using the circuit binary search tree over the whole run of the algorithm. (Each vertex turns from heavy to light at most once.)
  \item Perform a reversed breadth-first-search from all the heavy vertices simultaneously. We can assume that every vertex on the path is light (i.e., we find a ``closest'' heavy vertex reachable from $s$), and thus all its out-neighbors have already been discovered.
  That is, going backward from $S$ to $\bar{S}$, we use the out-edges of light vertices.
  From $\bar{S}$ to $S$, we use the circuit binary search tree.
  This takes $\tO(n)$ time, and we either find a heavy vertex reachable from $s$ (in which case we make progress by visiting at least $\sqrt{r}$ vertices in $S$), or we conclude that all heavy vertices are unreachable from $s$ (in which case $t$ is unreachable either).
\end{enumerate}

The number of iterations is bounded by $O(\sqrt{r})$ since we discover at least $\sqrt{r}$ unvisited vertices in $S$ every time.
The total running time of finding a single augmenting path is thus $\tO(n\sqrt{r})$.

Using the same parameters $\epsilon$ and $d$ as in \cite{blikstad21} to combine the three phases, we obtain the following matroid intersection algorithm in the dynamic-independence-oracle model.

\begin{theorem}
  For two matroids $\cM_1 = (\ground, \cI_1)$ and $\cM_2 = (\ground, \cI_2)$, it takes $\tO(nr^{3/4})$ time to obtain the largest $S \in \cI_1 \cap \cI_2$ with high probability in the dynamic-independence-oracle model.
  \label{thm:dynamic-matroid-intersection-ind-main}
\end{theorem}

\end{document}